\documentclass[11pt]{article}
\pdfoutput=1
\usepackage[utf8]{inputenc}
\usepackage{mathtools,enumitem,bbm,changepage,tocloft}

\usepackage[margin=1in]{geometry}

\usepackage[bottom]{footmisc} %
\usepackage[dvipsnames]{xcolor}

\usepackage{amsmath}
\usepackage{amsthm}
\usepackage{hyperref}
\usepackage{bookmark}%
\usepackage{cleveref}
\Crefname{algocf}{Algorithm}{Algorithms}
\usepackage{amsfonts}       %
\usepackage{booktabs}       %
\usepackage{subcaption}
\usepackage{blkarray}
\usepackage[ruled]{algorithm2e}
\usepackage[round]{natbib}

\newcommand{\initOneLiners}{%
    \setlength{\itemsep}{0pt}
    \setlength{\parsep }{0pt}
    \setlength{\topsep }{0pt}
      \usecounter{myLISTctr}
}
\newenvironment{OneLiners}[1][\ensuremath{\bullet}]
    {\begin{list}
        {#1}
        {\initOneLiners}}
    {\end{list}}

\newcommand{\op}{\ensuremath{\mathrm{op}}} %
\newcommand{\bR}{\ensuremath{\mathbb{R}}} %
\newcommand{\cD}{\ensuremath{\mathcal{D}}} %
\newcommand{\cP}{\ensuremath{\mathcal{H}}} %
\newcommand{\bE}{\ensuremath{\mathbb{E}}} %
\newcommand{\sps}[1]{\ensuremath{^{(#1)}}} %
\newcommand{\tw}{\ensuremath{\widetilde w}} %
\newcommand{\tD}{\ensuremath{\widetilde {\mathcal D}}} %
\newcommand{\tX}{\ensuremath{\widetilde X}} %
\newcommand{\cE}{\ensuremath{\mathcal E}} %
\newcommand{\err}{\ensuremath{\mathsf{err}}} %

\newcommand{\org}[1]{{\color{orange}#1}}

\renewcommand{\top}{{\ensuremath{\mathsf{T}}}}

\newtheorem{theorem}{Theorem}
\newtheorem{lemma}[theorem]{Lemma}
\newtheorem{claim}[theorem]{Claim}

\newtheorem{definition}{Definition}

\title{Robust Mean Estimation on Highly Incomplete Data with Arbitrary Outliers}

\author{Lunjia Hu
\thanks{Computer Science Department, Stanford University. 
Email: \texttt{lunjia@stanford.edu}.
Supported by NSF Award IIS-1908774 and a VMware fellowship.} 
\and 
Omer Reingold
\thanks{Computer Science Department, Stanford University.
Email: \texttt{reingold@stanford.edu}.
Supported in part by NSF Award IIS-1908774 and by Simons Foundation investigators award 689988.}}

\date{}

\allowdisplaybreaks

\begin{document}
\maketitle
\begin{abstract}
We study the problem of robustly estimating the mean of a $d$-dimensional distribution given $N$ examples, where most coordinates of every example may be missing and $\varepsilon N$ examples may be arbitrarily corrupted. Assuming each coordinate appears in a constant factor more than $\varepsilon N$ examples, we show algorithms that estimate the mean of the distribution with information-theoretically optimal dimension-independent error guarantees in nearly-linear time $\widetilde O(Nd)$. Our results extend recent work on computationally-efficient robust estimation to a more widely applicable incomplete-data setting.
\end{abstract}
\section{Introduction}
\label{sec:intro}
Imagine a program committee that wants to evaluate a collection of papers with a group of reviewers. As it is often impractical for a single reviewer to read all the papers, the common practice is to assign every paper only to a small sub-group of the reviewers. Assuming that all the reviewers are reliable, and the ratings they provide for a paper $p$ are independently drawn from a distribution $\cD_p$, one can estimate the mean of the distribution simply by taking the average rating over the sub-group. 
The estimation error can be reduced towards zero as the size of the sub-group grows.
In reality, however, not all ratings are reliably from $\cD_p$, and the empirical average often gives a bad estimate. 

In this work, we study robust estimation in such scenarios where we only have a partially-reliable dataset, which is also highly incomplete -- each reviewer only provides information about a small subset of the papers. 
These scenarios also appear naturally in other contexts such as crowdsourcing \citep{vuurens2011much} and peer-grading for online courses \citep{piech2013tuned,kulkarni2013peer}, 
where small pieces of data are aggregated from a large population that cannot be completely trusted.

Recently, there has been a line of work on robustly estimating the mean and covariance of a distribution given a partially-reliable dataset
\citep{diakonikolas2019recent,diakonikolas2019robust,lai2016agnostic,charikar2017learning,diakonikolas2017being,diakonikolas2018robustly,steinhardt2018resilience,cheng2019high,depersin2019robust,dong2019quantum,lei2020fast,cheng2020high,cherapanamjeri2020list,zhu2020robust,diakonikolas2020outlier,hopkins2020robust}.
Following the pioneering works by \citet{tukey1960survey} and \citet{huber1964robust} in robust statistics,
the dataset is modeled as 
initially i.i.d.\ examples $X\sps 1,\ldots,X\sps N \in\bR^d$ from the distribution, but $\varepsilon N$ of the examples are arbitrarily corrupted by an adversary.
The presence of an $\varepsilon$-fraction of adversarial corruption not only makes the empirical average a bad estimator for the mean, but also makes it information-theoretically impossible to reduce the estimation error towards zero even with infinitely many examples. For estimating the mean of a spherical Gaussian, the Tukey median \citep{tukey1975mathematics} achieves minimax-optimal error guarantees \citep{chen2018robust,zhu2020does}, but it is hard to compute in general \citep{johnson1978densest,diakonikolas2017statistical}; the linear-time-computable coordinate-wise median gives much worse $L_2$ errors that grow with the dimension.
Recent breakthroughs by \citet{diakonikolas2019robust} and \citet{lai2016agnostic} led to the best of both worlds: there are nearly-linear time algorithms achieving information-theoretically optimal error guarantees \citep{cheng2019high,depersin2019robust,dong2019quantum,lei2020fast}.

One limitation of these existing works is that they require observing 
all the $d$ coordinates of each example,
which is often impractical when the dimension is huge. In peer-grading, for example, a single student can only grade a very small fraction of the class. We thus ask: can one estimate the mean of a $d$-dimensional distribution given an $\varepsilon$-corrupted dataset in which most of the coordinates of every example are missing? 

If a particular coordinate is missing in every example, there is certainly no way to estimate that coordinate of the mean accurately. We thus assume that every coordinate appears in at least $\gamma N$ examples for $\gamma = \Omega(\varepsilon)$ so that the majority of the present coordinates are not corrupted. Even with this assumption, the incompleteness of the dataset brings additional challenges. The Tukey median, for example, is not well-defined in such an incomplete dataset. Natural approaches to eliminating the missing coordinates also fail as we show in \Cref{sec:fail}.
Nonetheless, we can still achieve optimal dimension-independent error guarantees in nearly-linear time as if there were no missing coordinates (see \Cref{sec:results} for our results).

Following previous works \citep{diakonikolas2017being,cheng2019high,dong2019quantum}, we focus on two classes of distributions: 1) sub-Gaussian distributions with equal variance along all directions, and 2) distributions with bounded variances along all directions (see \Cref{sec:distributions} for formal definitions).
Assuming that the present non-corrupted coordinates are drawn from such a distribution, our algorithms estimate the mean of the distribution by starting with a coarse guess obtained by taking the coordinate-wise median, and refining the guess iteratively. 
In order to apply existing algorithms which only work for complete datasets,
we first fill the missing coordinates with pre-determined values drawn i.i.d.\ from a mean-zero Gaussian distribution. In each iteration, we adjust the pre-determined values according to our current guess of the mean, and 
apply an existing algorithm to the resulting dataset to
compute the next guess.
We show that every new guess improves a potential function by a multiplicative factor, 
allowing us to
eventually reduce the estimation error to the information-theoretical limit.
The properties we need for the existing algorithm are
described in \Cref{thm:no-missing},
which we prove using the algorithm by \citet{dong2019quantum-arxiv} to obtain the best running time.

\subsection{Corruption Model}
\label{sec:model}
We specify our assumptions on how the corrupted and incomplete dataset is generated.
We assume that all the examples are drawn independently from ``nice'' distributions in a class $\cP$ before some of them are corrupted.
Specifically, we assume our dataset $X\sps 1,\ldots,X\sps N\in(\bR\cup\{*\})^d$ is generated by an adversary according to the following procedure, where $*$ is the placeholder for the missing coordinates:
\begin{OneLiners}
\item[1.] The adversary chooses a set of present (non-missing) coordinates $P\sps i\subseteq\{1,\ldots,d\}$ for every $i = 1,\ldots,N$.
\item[2.] The adversary chooses a $d$-dimensional distribution $\cD\sps i\in\cP$ for every $i = 1,\ldots,N$, and draws $X\sps i\in\bR^d$ independently from $\cD\sps i$. Every $\cD\sps i$ must have the same mean $\mu\in\bR^d$, which is also chosen by the adversary. 
\item[3.] The adversary modifies at most $\varepsilon N$ examples $X\sps i$ arbitrarily, potentially changing the corresponding $P\sps i$;
\item[4.] For every example $X\sps i$, its coordinates $X\sps i_j$ are concealed by $*$ whenever $j\notin P\sps i$.
\end{OneLiners}

In short, the adversary can choose for every example which coordinates are missing, and can corrupt an $\varepsilon$ fraction of the examples after the examples are drawn. We call a dataset generated this way an \emph{$\varepsilon$-corrupted dataset}. \Cref{fig:data} shows an example for $d = 4$ and $N = 7$.

\begin{figure}[h]
\begin{equation*}
\begin{blockarray}{ccccc@{\hspace{20pt}}@{\hspace{20pt}}ccccc}
\mu & 1.0 & 2.0 & 3.0 & 4.0 & \mu & 1.0 & 2.0 & 3.0 & 4.0 \\
\begin{block}{c[cccc]@{\hspace{20pt}}@{\hspace{20pt}}c[cccc]}
X\sps 1 & 1.2 & 1.8 & \org{2.9} & 4.0 &  
X\sps 1 & * & * & 2.9 & *
\\
X\sps 2 & \org{0.9} & 2.2 & \org{2.8} & \org{3.9} &  
X\sps 2 & 0.9 & * & 2.8 & 3.9
\\
X\sps 3 & 0.8 & 1.9 & 3.1 & \org{4.1} &  
X\sps 3 & * & * & * & 4.1
\\
X\sps 4 & \org{1.1} & 2.1 & 2.9 & \org{4.1} &  
X\sps 4 & \mathbf{0.5} & \mathbf{2.6} & * & \mathbf{5.0}
\\
X\sps 5 & 1.0 & 2.0 & 3.0 & 3.8 &  
X\sps 5 & * & * & * & * 
\\
X\sps 6 & 1.2 & \org{2.0} & \org{2.9} & 4.2 &  
X\sps 6 & * & 2.0 & 2.9 & * 
\\
X\sps 7 & \org{1.2} & \org{2.1} & 3.2 & 3.9 &  
X\sps 7 & 1.2 & 2.1 & * & *
\\
\end{block}
\end{blockarray}
\end{equation*}
\caption{The orange entries are the present coordinates chosen in step 1. The left matrix is the dataset after step 2. The right matrix is what we observe after the adversary has corrupted $X\sps 4$.}
\label{fig:data}
\end{figure}
For technical convenience, we allow the adversary to be randomized. Thus, once the adversary is fixed (as a randomized algorithm), the randomness of the dataset comes from the random draws from $\cD\sps i$ and the inherent randomness of the adversary. 
It is often more convenient to consider the present entries coordinate-wise rather than example-wise, so we define $\Gamma_j:=\{i\in\{1,\ldots,N\}:j\in P\sps i\}$ as the set of examples with coordinate $j$ present.
The adversary may sometimes choose $\Gamma_j$ to be the empty set, in which case we have no hope to estimate $\mu_j$. Therefore, we measure the performance of our algorithms only on \emph{$\gamma$-complete datasets}, i.e., datasets with $|\Gamma_j|\geq \gamma N$ for all coordinates $j$.

\subsection{Our Results}
\label{sec:results}
Given an $\varepsilon$-corrupted dataset generated by the procedure in \Cref{sec:model}, we show nearly-linear time algorithms that compute good estimates for the true mean $\mu$ with high probability whenever the dataset is $\gamma$-complete.
Specifically, we prove the following theorems for different classes of distributions defined more formally in \Cref{sec:distributions}:
\begin{theorem}
\label{thm:Gaussian}
Given an $\varepsilon$-corrupted dataset 
$X\sps 1,\ldots,X\sps N\in(\bR\cup\{*\})^d$
generated with $\cP$ being the class
$\cP_1(\eta)$ of $1$-sub-Gaussian distributions with covariance $\eta^2 I$. Assume $C_0\varepsilon\leq \gamma \leq 1$, $\delta \in (0, 1/2)$, and 
$\gamma N \geq \Omega(\log(d/\delta))$,
where $C_0 > 2$ is an absolute constant. 
There is an $\widetilde O(Nd\log(1/\delta))$-time algorithm with the following property:
the event that the dataset is $\gamma$-complete but the algorithm does \emph{not} output an estimate of the true mean $\mu$ up to $L_2$ error 
$\err$
happens with probability at \emph{most} $\delta$, where
\begin{equation}
\label{eq:main-error-1}
\err = O\Big(\sqrt{(d + \log (1/\delta))/\gamma N + (\varepsilon/\gamma)^2\log(\gamma/\varepsilon)}\Big).
\end{equation}
In particular, when 
\begin{equation}
\label{eq:main-sample-1}
\gamma N \geq \Omega\left(\frac{\gamma^2(d + \log(1/\delta))}{\varepsilon^2\log(\gamma/\varepsilon)}\right),
\end{equation}
the error bound simplifies to 
\begin{equation}
\label{eq:main-simple-error-1}
\err = O\Big((\varepsilon/\gamma)\sqrt{\log(\gamma/\varepsilon)}\Big).
\end{equation}
\end{theorem}
\begin{theorem}
\label{thm:2nd-moment}
Given an $\varepsilon$-corrupted dataset 
$X\sps 1,\ldots,X\sps N\in(\bR\cup\{*\})^d$
generated with $\cP$ being the class
$\cP_2$ of distributions with covariances $\Sigma \preceq I$. Assume $C_0\varepsilon\leq \gamma \leq 1$, $\delta\in (0, 1/2)$ and 
$\gamma N \geq \Omega(\log(d/\delta))$,
where $C_0 > 2$ is an absolute constant.
There is an $\widetilde O(Nd\log(1/\delta))$-time algorithm with the following property:
the event that the dataset is $\gamma$-complete but the algorithm does \emph{not} output an estimate of the true mean $\mu$ up to $L_2$ error 
$\err$
happens with probability at \emph{most} $\delta$, where
\begin{equation}
\label{eq:main-error-2}
\err = O\Big(\sqrt{d\log(d/\delta)/\gamma N + \varepsilon/\gamma}\Big).
\end{equation}
In particular,
when 
\begin{equation}
\label{eq:main-sample-2}
\gamma N \geq \Omega\left(\frac{\gamma d\log(d/\delta)}{\varepsilon}\right),
\end{equation}
the error bound simplifies to
\begin{equation}
\label{eq:main-simple-error-2}
\err = O\Big(\sqrt{\varepsilon/\gamma}\Big).
\end{equation}
\end{theorem}

\paragraph{Tightness of bounds.}
Our error guarantees \eqref{eq:main-simple-error-1} and \eqref{eq:main-simple-error-2} are independent of the dimension $d$. If we focus on the ``hardest'' regime $\eta \ge \Omega(1)$ in \Cref{thm:Gaussian}, the error guarantees \eqref{eq:main-simple-error-1} and \eqref{eq:main-simple-error-2} are information-theoretically optimal up to constants even with infinitely many examples. This is most easily seen by a reduction from 
the complete-data setting. Given an $\varepsilon'$-corrupted dataset of $N'$ examples without missing coordinates, we can simply add trivial examples with only missing coordinates to create an $\varepsilon$-corrupted $\gamma$-complete dataset 
of size $N$
for $\varepsilon = \varepsilon'\gamma$ and $N = N'/\gamma$. The optimal error guarantees $O\big(\varepsilon' \sqrt{\log (1/\varepsilon')}\big)$ and $O\big(\sqrt{\varepsilon'}\big)$ for complete data\footnote{
Both error guarantees 
$O\big(\varepsilon' \sqrt{\log (1/\varepsilon')}\big)$ and $O\big(\sqrt{\varepsilon'}\big)$ 
are optimal even for one-dimensional data. 
The optimality of $O\big(\varepsilon' \sqrt{\log (1/\varepsilon')}\big)$ 
is shown by \citet[Lemma 19]{diakonikolas2017robustly-arxiv}.
The optimality of $O\big(\sqrt{\varepsilon'}\big)$ 
follows easily from the indistinguishability of 
the following two distributions on $\varepsilon'$-corrupted datasets: 
the singleton distribution at 0, and 
the distribution over $\big\{0, 1/\sqrt{\varepsilon'}\big\}$ 
where the probability of getting $1/\sqrt{\varepsilon'}$ 
is $0.1\varepsilon'$.
}
then translate exactly to 
\eqref{eq:main-simple-error-1} and \eqref{eq:main-simple-error-2}, respectively. 
Similarly, in order to achieve the error bounds \eqref{eq:main-simple-error-1} and \eqref{eq:main-simple-error-2}, our sample complexities \eqref{eq:main-sample-1} and \eqref{eq:main-sample-2} are optimal up to logarithmic factors.
It is an interesting open question to achieve \emph{sub-Gaussian rates} in \Cref{thm:2nd-moment}, i.e., to replace $d\log (d/\delta)$ by $d + \log(1/\delta)$ (see \citet{depersin2019robust,lei2020fast,diakonikolas2020outlier} for some recent progress in the complete data setting).

\paragraph{Assumptions on distributions.}
By scaling the entire dataset, 
for any known $\sigma^2$,
\Cref{thm:Gaussian} generalizes to $\sigma^2$-sub-Gaussian distributions and 
\Cref{thm:2nd-moment} generalizes to distributions with covariances $\Sigma \preceq \sigma^2 I$.
If one wants to apply \Cref{thm:Gaussian} to distributions whose covariances are not $\eta^2 I$, but some other known diagonal matrix $\Sigma$, one only needs to scale each coordinate properly to make the covariance matrix equal to $I$. However, if $\Sigma$ is not diagonal, i.e., the coordinates are correlated with each other, it is less clear how one could perform such a transformation. This is in contrast to the case without missing coordinates where any known covariance can be reduced to identity. We leave the interesting problem of handling correlated coordinates better in corrupted incomplete datasets to future work.
Also, if the covariance is unknown to us, the adversary can easily hide the correlation between a pair of coordinates in our corruption model by never creating an example with both coordinates present, so robust covariance estimation on an incomplete dataset requires different assumptions, which would also be an interesting topic for future work.

\paragraph{How presence/missingness is determined.}
In our corruption model,
the locations of the present coordinates
are chosen by the adversary in step 1. 
In some applications, however, the algorithm itself can choose the locations:
in the paper reviewing setting, 
the program committee can decide the assignment of papers to reviewers. 
In such cases, we can effectively eliminate the missing coordinates
if we choose their locations appropriately (see \Cref{sec:hashing} 
where we also show that 
choosing the locations uniformly at random does \emph{not} give 
the ``easiest'' instances).
Nevertheless, our results hold even when the locations are adversarially chosen, 
applying more broadly to cases where the locations of the present coordinates are restricted by, for example,
conflict of interest and the lack of expertise of a reviewer in certain research areas.
In the other direction, one may consider a stronger adversary where 
the first two steps are swapped:
the adversary decides which coordinates are present \emph{after} seeing the entire data matrix. In this case,
the estimation error would grow with the dimension 
even without corruption:
for example,
assuming every entry in the matrix is i.i.d.\
from $\{0,a\}$ 
where the probability of getting $0$ is slightly higher than $\gamma$,
the adversary can choose to include $i$ in $\Gamma_j$ only when $X\sps i_j = 0$, making it impossible to accurately estimate the mean $\mu$ which clearly depends on $a$
(see \citet{liu2020robust} for more results on datasets where missing entries are chosen by such stronger adversaries).

\subsection{Other Related Works}
Our work is close in spirit to \citet{steinhardt2016avoiding}, who also considered corrupted and incomplete datasets. They made generally weaker assumptions on the dataset: the majority of the examples may be corrupted, fewer coordinates are observed from each example, and different uncorrupted examples may come from different distributions as long as the discrepancies between pairs of coordinates are mostly preserved. 
Consequently, their goal is also weaker, which is to identify $\beta d$ coordinates with approximately the highest average mean given the additional ability to obtain reliable values of a small number of coordinates.

Besides mean and covariance estimation, many other problems have been studied with corrupted data, such as principal component analysis 
\citep{candes2011robust,chandrasekaran2011rank,xu2010principal}, 
learning graphical models \cite{diakonikolas2016robust}, linear regression 
\citep{bhatia2015robust,bhatia2017consistent,diakonikolas2019efficient}, 
sparse estimation
\citep{balakrishnan2017computationally,liu2020high,liu2019high},
and learning discrete distributions \citep {qiao2018learning,chen2020learning}. 
More general models of data corruption are considered by 
\citet{zhu2019generalized},
as well as
in the robust hypothesis testing literature
\citep{huber1973minimax,verdu1984minimax,levy2008robust,gul2017robust,gul2017minimax,gao2018robust}.

Early influential works on statistical inference using incomplete data include \citet{rubin1976inference}, \citet{dempster1977maximum}, and \citet{rubin1979illustrating}.
Data incompleteness is also a general theme of more recent research. For instance, matrix completion is the problem of recovering missing entries in a data matrix, usually under the assumption that the matrix is (approximately) low-rank \citep{candes2009exact,candes2010power,candes2010matrix}. Trace reconstruction is the problem of reconstructing a sequence from its subsequences (traces) \citep{batu2004reconstructing,holenstein2008trace}.

Even without data corruption or incompleteness, it is non-trivial to estimate the mean of a distribution with \emph{optimal} error rates w.r.t.\ a small given failure probability bound $\delta$, especially when the distribution is \emph{heavy-tailed}. \citet{lugosi2019sub} first showed that achieving \emph{sub-Gaussian rates} is possible assuming only bounded second moment (see \citet{lugosi2019mean} for a more comprehensive survey).
\section{Preliminaries}
\label{sec:preliminaries}
\subsection{Classes of Distributions}
\label{sec:distributions}
For $\sigma\geq 0$,
we say a distribution $\cD$ over $\bR^d$ with mean $\mu$ is $\sigma^2$-sub-Gaussian if for all vectors $v$, it holds that $\bE_{X\sim\cD}[\exp((X - \mu)^\top v)]\leq \exp(\sigma^2\|v\|_2^2/2)$. 
Given $\eta\geq 0$, class $\cP_1(\eta)$ consists of $1$-sub-Gaussian distributions with covariance $\eta^2 I$. Class $\cP_2$ consists of distributions with covariances $\Sigma \preceq I$.
It is easy to check that any $\sigma^2$-sub-Gaussian distribution has covariance $\Sigma \preceq \sigma^2 I$, so $\cP_1(\eta)$ is non-empty only when $\eta \le 1$. We thus assume $\eta \le 1$ throughout the paper. The following simple property of $\cP_1(\eta)$ and $\cP_2$ turns out to be handy in the analysis of our algorithm -- both $\cP_1(\eta)$ and $\cP_2$ are closed under coordinate-wise composition
(see \Cref{sec:proof-composition} for proof):
\begin{claim}
\label{claim:composition}
Let $\cP$ be $\cP_1(\eta)$ or $\cP_2$. 
Given $m$ distributions $\cD\sps 1,\ldots,\cD\sps m\in\cP$ all of which have the same mean $\mu$,
we define as follows their coordinate-wise composition $\cD$ 
w.r.t.\ a partition $J\sps 1,\ldots,J\sps m$ of the coordinates
$\{1,\ldots,d\}$.
We first draw $X\sps 1,\ldots,X\sps m$ independently from $\cD\sps 1,\ldots,\cD\sps m$, respectively,
and then construct a variable $X$ so that it agrees with $X\sps i$
on the coordinates in $J\sps i$ for all $i$.
We define the composition $\cD$ to be the distribution of $X$. 
Then, $\cD$ belongs to $\cP$ and it has mean $\mu$.
\end{claim}

\subsection{Other Notations}
For any $G\subseteq\{1,\ldots,N\}$, we define $\Delta_G$ as the set of weights $(w\sps 1,\ldots,w\sps N)$ satisfying $\forall i,w\sps i\geq 0$, $\sum_{i=1}^Nw\sps i= 1$, and $\forall i\notin G, w\sps i=0$.
For any $p\in(0,|G|]$, define $\Delta_{G,p}$ as the set of weights in $\Delta_G$ satisfying $\forall i,w\sps i\leq \frac 1p$. When $p$ is an integer, $\Delta_{G,p}$ is the convex hull of uniform distributions over size-$p$ subsets of $G$. 
We use $\|\cdot \|_2$ and $\|\cdot \|_\op$ to denote the vector $L_2$ norm and its induced matrix operator norm (a.k.a.\ the spectral norm), respectively.
\section{Estimation under Deterministic Conditions}
\label{sec:deterministic}
Our goal is to extract useful information from a corrupted and highly incomplete data matrix:
a fraction of the examples are controlled by an adversary, and most of the entries in our data matrix are missing. Somewhat surprisingly, given the terrible condition of the data matrix, some important structures are still preserved. In this section, we formally characterize the useful structures by a set of deterministic conditions, and show that we can perform efficient mean estimation as long as these deterministic conditions hold, regardless of how the data matrix was generated. In the next section (\Cref{sec:concentration}), we prove these deterministic conditions assuming the dataset is generated by the procedure in \Cref{sec:model}, and thus establish our main theorems.

\subsection{A General Result Assuming No Missing Coordinates}
When there are no missing coordinates, the deterministic conditions 
sufficient for performing robust mean estimation 
have been identified by many previous works \citep{diakonikolas2017being,steinhardt2018resilience,cheng2019high}. 
We state such a set of conditions by defining the goodness property of a dataset (\Cref{def:no-missing}), and show that this property is sufficient for robust mean estimation in \Cref{thm:no-missing}, which summarizes the results by \citet{dong2019quantum-arxiv}.

\begin{definition}[summarized based on \citet{cheng2019high}]
\label{def:no-missing}
We say a dataset $X\sps 1,\ldots,X\sps N\in\bR^d$
is $(\varepsilon,\eta,\beta)$-good with respect to $\mu\in\bR^d$
if there exists $G\subseteq\{1,\ldots,N\}$ with the following properties:
\begin{OneLiners}
\item $|G|\geq (1 - \varepsilon)N$;
\item For all $w\in \Delta_{G,(1 - 3\varepsilon)N}$,
\begin{align*}
& \big\|\sum_{i\in G}w\sps i(X\sps i - \mu)\big\|_2 \leq \beta\sqrt{\varepsilon},\ \text{and}\\
& \big\|\big(\sum_{i\in G}w\sps i(X\sps i - \mu)(X\sps i - \mu)^\top\big) - \eta^2 I\big\|_\op\leq \beta^2.
\end{align*}
\end{OneLiners}
\end{definition}

\begin{theorem}[implicit in \citet{dong2019quantum-arxiv}]
\label{thm:no-missing}
There are absolute constants $c_1\in (0,1/3)$ and $C_2 > 0$ with the following property.
Suppose that we have an $(\varepsilon, \eta, \beta)$-good dataset 
$X\sps 1,\ldots,X\sps N\in\bR^d$
with respect to some unknown $\mu$,
and assume that 
$0 \leq \varepsilon \leq c_1, 
\beta^2\geq \Omega(\eta^2\varepsilon\log (1/\varepsilon)),\delta\in(0, 1/2)$. 
There is an $\widetilde O(Nd\log (1/\delta))$-time algorithm computing an estimate $\nu$
which, with probability at least $1 - \delta$, satisfies that
\begin{equation*}
\|\nu - \mu\|_2^2\leq C_2\beta^2\varepsilon.
\end{equation*}
\end{theorem}
We defer the proof of \Cref{thm:no-missing} to \Cref{sec:QUE}.

\subsection{Improving the Guess When Coordinates are Missing}
\label{sec:determinisitic-missing}
When our examples contain missing coordinates, the $L_2$ norm and the operator norm in \Cref{def:no-missing} are not well-defined. We address this issue by filling the missing coordinates according to a guess $\nu$ of the true mean $\mu$, and prove \Cref{lm:goodness} showing that
the estimation error of $\nu$ controls
the goodness of the resulting complete dataset, 
on which we can compute our next guess $\nu'$ using the algorithm from \Cref{thm:no-missing}.
In \Cref{thm:deterministic}, we bound the estimation error of the new guess $\nu'$ directly by the estimation error of the old guess $\nu$ together with a goodness measure  for the original \emph{incomplete} dataset, which we define in \Cref{def:missing}.
This ensures that our guess becomes accurate enough 
for proving our main theorems after $O(\log N)$ iterations, 
as we will show in the next section 
(\Cref{sec:concentration}).

Given an incomplete dataset $X\sps 1,\ldots, X\sps N$, 
we first fill the missing coordinates in a simple pre-determined way, 
which we will specify in \Cref{sec:concentration}.
For now, it suffices to assume that we have a \emph{completed dataset}
$X\sps {1,0},\ldots,X\sps{N,0}\in\bR^d$ that agrees with
the original dataset on the present entries.
Given a guess $\nu$, 
we define the \emph{adjusted dataset} $X\sps {1,\nu},\ldots,X\sps {N,\nu}$ by 
\begin{equation}
\label{eq:adjust}
X\sps {i,\nu}_j = \left\{\begin{array}{ll} 
X\sps {i,0}_j = X\sps i_j, & \textup{if} ~ X\sps i_j \neq *,\\
X\sps {i,0}_j + \nu_j, & \textup{otherwise}.
 \end{array}\right.
\end{equation}
We define the deterministic conditions on the completed dataset $X\sps{1,0},\ldots,X\sps{N,0}$:
\begin{definition}
\label{def:missing}
Given an incomplete dataset $X\sps 1,\ldots, X\sps N\in(\bR\cup\{*\})^d$, 
recall that $\Gamma_j$ is the set of examples
with coordinate $j$ present, i.e., 
$X\sps i_j \neq * \Longleftrightarrow i\in \Gamma_j$.
We say a corresponding completed dataset $X\sps{1,0},\ldots,X\sps {N,0}\in\bR^d$  
is $(\varepsilon,g_*,\eta,\beta)$-good with respect to $\mu$
if there exists $G\subseteq \{1,\ldots,N\}$ with the following properties:
\begin{OneLiners}
\item $|G|\geq (1 - \varepsilon)N$;
\item $g_j:=|G_j|/|G|\geq g_*$ for all $j = 1,\ldots,d$, where $G_j:=G\cap \Gamma_j$;
\item For all $w\in \Delta_{G,(1 - 3\varepsilon)N}$, 
\begin{align*}
& \big\|\sum_{i\in G}w\sps i(X\sps{i,\mu}- \mu)\big\|_2 \leq \beta\sqrt\varepsilon,\\
& \big\|\big(\sum_{i\in G}w\sps i(X\sps{i,\mu} - \mu)(X\sps {i,\mu} - \mu)^\top\big) -  \eta^2 I\big\|_\op \leq  \beta^2,\\
& \big\|\sum_{i\in G_j}w\sps i(X\sps{i,\mu} - \mu)\big\|_2 \leq \beta\sqrt\varepsilon,\quad \forall j = 1,\ldots,d.
\end{align*}
\end{OneLiners}
\end{definition}
We differentiate \Cref{def:missing} from \Cref{def:no-missing}
by the number of parameters: \Cref{def:missing} has an additional
parameter $g_*$ suggesting the incomplete nature of the original dataset.
It is clear that a completed dataset $X\sps{1,0},\ldots,X\sps{N,0}$ being $(\varepsilon,g_*,\eta,\beta)$-good with respect to $\mu$ (\Cref{def:missing}) implies that the adjusted dataset $X\sps{1,\mu},\ldots,X\sps{N,\mu}$ being $(\varepsilon,\eta,\beta)$-good (\Cref{def:no-missing}). 
However, this adjusted dataset
depends on the unknown $\mu$,
so we cannot apply \Cref{thm:no-missing} directly to it.
The following lemma (proved in \Cref{sec:proof-goodness})
shows that if we instead use a guess $\nu$ 
to adjust the dataset, the goodness of the resulting dataset
is controlled by the estimation error of $\nu$
measured by
\begin{equation*}
\|\nu - \mu\|_g := \Big(\sum_{j = 1}^dg_j(\nu_j - \mu_j)^2\Big)^{1/2}.
\end{equation*}
\begin{lemma}
\label{lm:goodness}
There are absolute constants $C_3>1,C_4>1$ with the following property.
Suppose we have an $(\varepsilon,g_*, \eta,\beta)$-good completed dataset $X\sps {1,0},\ldots,X\sps{N,0}$ with respect to $\mu$. 
Assume $5\varepsilon\leq g_* \leq 1$. 
Suppose $G,g_1,\ldots,g_d$ satisfy the conditions in \Cref{def:missing}.
Let $\nu$ be a guess with $\|\nu - \mu\|_g \leq \rho$. Then the adjusted dataset $X\sps {1,\nu},\ldots,X\sps {N,\nu}$ is $(\varepsilon,\eta,\beta')$-good with respect to $\mu'$ where
$
\mu'_j := g_j \mu_j + (1 - g_j) \nu_j$
and 
$
(\beta')^2 := C_3\beta^2 + C_4\rho^2.
$
\end{lemma}
Given the above goodness property of the adjusted dataset
$X\sps{1,\nu},\ldots, X\sps{N,\nu}$,
we can apply the algorithm in \Cref{thm:no-missing} to it
to compute an improved guess $\nu'$:
\begin{theorem}
\label{thm:deterministic}
There are absolute constants $C_5 \geq 5,C_6 > 0,c_7\in (0,1)$ with the following property.
Suppose we have an $(\varepsilon, g_*, \eta, \beta)$-good completed dataset 
$X\sps {1,0},\ldots,X\sps{N,0}\in\bR^d$ with respect to $\mu$,
where $C_5\varepsilon\leq g_*\leq 1, \eta\geq 0, \beta^2\geq \Omega(\eta^2\varepsilon\log(1/\varepsilon))$.
Suppose the $\mu,G,g_1,\ldots,g_d$ satisfying the conditions in \Cref{def:missing}
are all unknown.
Given a current guess $\nu\in\bR^d$,
an upper bound $\rho$ for $\|\nu - \mu\|_g$,
and a desired failure probability bound $\delta\in (0, 1/2)$,
there is an $\widetilde O(Nd\log(1/\delta))$-time algorithm 
computing a new guess $\nu'$ which,
with probability at least $1 - \delta$,
satisfies
\begin{equation*}
\|\nu' - \mu\|_g^2\leq C_6\beta^2\varepsilon + (1 - c_7g_*)\rho^2.
\end{equation*}
\end{theorem}
\begin{proof}
We choose $C_5 \geq 1/c_1$ for the $c_1$ in \Cref{thm:no-missing}. 
This ensures that $\varepsilon \leq c_1$.
We apply the algorithm in \Cref{thm:no-missing} to $X\sps{1,\nu},\ldots,X\sps{N,\nu}$, which is $(\varepsilon, \eta,\beta')$-good with respect to $\mu'$ by \Cref{lm:goodness}, where
\begin{align}
(\beta')^2 & = C_3\beta^2 + C_4\rho^2, \label{eq:beta'}\\
\mu_j' & =g_j\mu_j + (1 - g_j)\nu_j,\quad\textup{for} ~ j = 1,\ldots,d.\label{eq:mu'}
\end{align}
\Cref{thm:no-missing} computes an estimate $\nu'$ for $\mu'$, and as we show below, $\nu'$ is also a good guess for $\mu$ itself.
Define $a:=\nu - \mu$ and $a':=\nu' - \mu$. By Cauchy-Schwarz, for all $j = 1,\ldots,d$, we have
\begin{equation*}
\big((a_j' - (1 - g_j)a_j)^2 + (1 - g_j)g_ja_j^2\big)(g_j + (1 - g_j))\geq g_j(a_j')^2.
\end{equation*}
Summing up over $j = 1,\ldots,d$, this implies that 
\begin{equation}
\label{eq:cauchy-decomposition}
\|\nu' - \mu\|_g^2 = \sum_{i=1}^dg_j(a_j')^2\leq \|\nu' - \mu'\|_2^2 + \sum_{j=1}^d(1 - g_j)g_j(\nu_j - \mu_j)^2,
\end{equation}
where we used the fact that $\nu'_j - \mu'_j = (\nu_j' - \mu_j) - (\mu_j' - \mu_j) =  a_j' - (1 - g_j)a_j$ because of \eqref{eq:mu'}.
The guarantee of \Cref{thm:no-missing} is
\begin{equation}
\label{eq:guarantee-max}
\|\nu' - \mu'\|_2^2\leq C_2(\beta')^2\varepsilon.
\end{equation}
Plugging \eqref{eq:guarantee-max} into \eqref{eq:cauchy-decomposition}, we have:
\begin{align*}
\|\nu' - \mu\|_g^2 & \leq C_2(C_3\beta^2 + C_4\rho^2)\varepsilon + (1 - g_*)\rho^2  \tag{by \eqref{eq:beta'} and $g_j\geq g_*$}\\
& = C_2C_3\beta^2\varepsilon + (1 - g_* + C_2C_4\varepsilon)\rho^2\\
& \leq C_6\beta^2\varepsilon + (1 - c_7g_*)\rho^2. \tag{$g_*\geq C_5\varepsilon$}
\end{align*}
\end{proof}
\section{Finite Sample Concentration}
\label{sec:concentration}
In this section, we prove \Cref{thm:Gaussian,thm:2nd-moment} by 
showing that the deterministic conditions in \Cref{def:missing} are indeed satisfied by $\varepsilon$-corrupted $\gamma$-complete datasets generated 
according to \Cref{sec:model}
as long as we fill the missing entries properly
to obtain the completed dataset:
when $\cP = \cP_1(\eta)$, we fill the missing entries by independent $\mathcal N(0,\eta^2)$ variables,
and when $\cP = \cP_2$, we fill the missing entries simply by zeros.
In either case, we fill the missing entries using a mean-zero distribution
from $\cP$,
so by \Cref{claim:composition}, 
when we adjust the completed dataset by the true mean $\mu$,
it would
look like a dataset generated from $\cP$
initially without missing entries.
Specifically, we have the following two lemmas proved in \Cref{sec:subgaussian-goodness,sec:bounded-variance-goodness}, respectively.
\begin{lemma}
\label{lm:subgaussian-goodness}
Assume an $\varepsilon$-corrupted dataset
$X\sps 1,\ldots,X\sps N\in\bR^d$
is generated with $\cP = \cP_1(\eta)$.
Assume $10\varepsilon\leq \gamma \leq 1,
\delta \in (0, 1/2)
$. 
The event that the dataset is $\gamma$-complete 
but the corresponding completed dataset
filled by independent $\mathcal N(0,\eta^2)$ variables
is not $(\varepsilon',g_*,\eta,\beta)$-good
happens with probability at most $\delta$,
where
$\varepsilon' = \varepsilon,g_* = 0.9\gamma$, and
\[
	\beta = \Theta\Big(\sqrt{(d + \log (1/\delta))/N\varepsilon + \varepsilon\log(1/\varepsilon)}\Big).
\]
\end{lemma}
\begin{lemma}
\label{lm:bounded-variance-goodness}
Assume an $\varepsilon$-corrupted dataset 
$X\sps 1,\ldots,X\sps N\in\bR^d$
is generated with $\cP = \cP_2$.
Assume $11\varepsilon\leq \gamma \leq 1,
\delta\in(0,1/2)$,
and $N \varepsilon \geq \Omega(\log(d/\delta))$.
The event that the dataset is $\gamma$-complete 
but the corresponding completed dataset 
filled by zeros
is not $(\varepsilon',g_*,\eta,\beta)$-good
happens with probability at most $\delta$,
where
$\varepsilon' = 1.1\varepsilon, g_* = 0.9\gamma,
\eta = 0,$ and 
\[
\beta = \Theta\Big(\sqrt{d\log (d/\delta)/N\varepsilon + 1}\Big).
\]
\end{lemma}
The goodness properties guaranteed by these lemmas allow us 
to apply \Cref{thm:deterministic} to improve our guess.
We choose our initial guess as the coordinate-wise median
with missing entries ignored,
whose estimation error is bounded by the following claim (proved in \Cref{sec:proof-median}):
\begin{claim}
\label{claim:coordinate-wise-median}
Assume an $\varepsilon$-corrupted dataset
$X\sps 1,\ldots,X\sps N\in\bR^d$
is generated 
with $\cP$ being $\cP_1(\eta)$ or $\cP_2$.
Assume $7\varepsilon\leq \gamma \leq 1,\delta\in (0, 1/2), \gamma N \geq \Omega(\log (d/\delta))$.
The event that the dataset is $\gamma$-complete
but the coordinate-wise median $\nu$ does not satisfy 
$\|\nu - \mu\|_2\leq O\big(\sqrt d\big)$
happens with probability at most $\delta$.
\end{claim}

We are now ready to summarize our algorithm and prove our main theorems:
\begin{algorithm}[h]
\SetKwInOut{Input}{input}\SetKwInOut{Output}{output}
\Input{An $\varepsilon$-corrupted $\gamma$-complete dataset $X\sps 1,\ldots,X\sps N$ generated with $\cP$ being $\cP_1$ or $\cP_2$.}
\Output{An estimate $\nu$ of the mean $\mu$.}
Initialize $\nu$ as the coordinate-wise median\;
Fill the missing entries as described to obtain a completed dataset
$X\sps {1,0},\ldots,X\sps {N,0}$\;
Define $C_6, c_7$ as in \Cref{thm:deterministic}\;
Define $\varepsilon',g_*,\eta,\beta$ as in 
\Cref{lm:subgaussian-goodness} or \Cref{lm:bounded-variance-goodness} depending on 
whether $\cP = \cP_1(\eta)$ or $\cP = \cP_2$\;
Initialize $\rho$ as the upper bound for $\|\nu - \mu\|_2$ in \Cref{claim:coordinate-wise-median}\;
\For{$t = 1,\ldots,O(\log N)$}
{
Run the algorithm in \Cref{thm:deterministic} on the completed dataset with current guess $\nu$,
upper bound $\rho$,
and parameters $\varepsilon',g_*,\eta,\beta$
to compute the next guess $\nu'$\;
$\nu\gets\nu'$\;
$\rho\gets (C_6\beta^2\varepsilon' + (1 - c_7g_*)\rho^2)^{1/2}$\;
}
 \caption{Robust mean estimation on incomplete data with arbitrary outliers}
 \label{alg}
\end{algorithm}

\begin{proof}[Proof of \Cref{thm:Gaussian,thm:2nd-moment}]
We will show a pre-processing procedure in the next section that 
allows us to assume w.l.o.g.\ that $\gamma$ is lower-bounded by
a small positive constant, i.e., $\gamma \geq \Omega(1)$. 

Our assumption $\gamma N \geq \Omega(\log(d/\delta))$ implies that $\log(d/\delta)/\gamma N = O(1)$. We can then assume w.l.o.g.\ that $\varepsilon \geq \Omega(\log(d/\delta)/\gamma N)$ in \Cref{thm:2nd-moment} because for smaller $\varepsilon$, $d\log(d/\delta)/\gamma N$ dominates $\varepsilon/\gamma$ in the error guarantee of \Cref{thm:2nd-moment}. The assumption $N\varepsilon \geq \Omega(\log(d/\delta))$ in \Cref{lm:bounded-variance-goodness} is now satisfied.

We prove that \Cref{alg} satisfies the conditions of \Cref{thm:Gaussian,thm:2nd-moment}.
The goodness property of the completed dataset $X\sps{1, 0},\ldots, X\sps{N, 0}$ is guaranteed by \Cref{lm:subgaussian-goodness} and \Cref{lm:bounded-variance-goodness}.
By \Cref{thm:deterministic}, in each iteration of \Cref{alg}, when the current guess $\nu$ satisfies $\|\nu - \mu\|_g\leq \rho$,
the next guess $\nu'$
satisfies $\|\nu' - \mu\|_g^2\leq (\rho')^2:=
C_6\beta^2\varepsilon' + (1 - c_7g_*)\rho^2$ , i.e.,
\begin{equation}
\label{eq:potential}
(\rho')^2 - \frac{C_6\beta^2\varepsilon'}{c_7g_*} = (1 - c_7g_*)\Big(\rho^2 - \frac{C_6\beta^2\varepsilon'}{c_7g_*}\Big).
\end{equation}
This shows a multiplicative decrease in a potential function.
Since the initial $\rho^2$ is $O(d)$ given by
\Cref{claim:coordinate-wise-median} and $g_* = 0.9\gamma \geq \Omega(1)$,
in $O(\log(1 + d/\beta^2\varepsilon')) = O(\log N)$ iterations
we obtain a guess $\nu^*$
with $\|\nu^* - \mu\|_2^2 \leq \|\nu^* - \mu\|_g^2/g_* \leq O(\beta^2\varepsilon')$, 
giving the desired error bounds for \Cref{thm:Gaussian,thm:2nd-moment}
after we plug in the values of $\varepsilon',\beta$
from \Cref{lm:subgaussian-goodness} and \Cref{lm:bounded-variance-goodness}, respectively.
The running time and the success probability of each iteration of \Cref{alg}
are given by \Cref{thm:deterministic}.
\end{proof}
\section{Simple but Failed Attempts and Pre-processing via Hashing}

\label{sec:fail} 
We discuss why two simpler methods (stacking and hashing)
cannot directly reduce the robust mean estimation problem
from the incomplete-data setting
to the complete-data setting.
We also show that we can use the hashing method
as the pre-processing step
in the proof of our main theorems
in order to assume
w.l.o.g.\ that $\gamma$ is at least a constant.

\subsection{Stacking}
In a $\gamma$-complete dataset, each coordinate is present in at least $\gamma N$ examples. Thus, one may simply ignore the missing entries and stack the examples together to form a $\gamma N \times d$ matrix. For example, the dataset given in \Cref{fig:data} becomes
\begin{equation*}
\begin{bmatrix}
0.9 & \mathbf{2.6} & 2.9 & 3.9 \\
\mathbf{0.5} & 2.0 & 2.8 & 4.1 \\
1.2 & 2.1 & 2.9 & \mathbf{5.0}
\end{bmatrix}.
\end{equation*}
Note that the corrupted entries (shown in bold) propagate to all the three resulting examples. In general, stacking breaks the fact that all the corruption happens in a small fraction of the examples, 
and this fact is essential for the $L_2$ error not to grow with the dimension.
\subsection{Hashing}
\label{sec:hashing}
Stacking fails partly because different examples in the new dataset may depend on the same, possibly corrupted example in the original dataset. 
We thus hope that different new examples depend on \emph{disjoint} subsets of the original examples.
Imagine that step 1 of the corruption procedure in \Cref{sec:model} is performed \emph{not} by the adversary; instead, we have the freedom to choose which entries are present as long as no example has too many present coordinates,
and the adversary is not allowed to modify $P\sps i$ in step 3.
In this case,
we can partition the $N$ examples into $N'\approx \gamma N$ groups of similar sizes, and
then for each group and each coordinate $j$, we choose an example in that group to have coordinate $j$ present.
In this way, the examples in each group can be combined into a new example \emph{without missing coordinates}
(see \Cref{fig:partition}).
The $j$-th coordinate of the new example equals to $X\sps i_j$ where $X\sps i$ is an original example in the group with its $j$-th coordinate present.
Hence, the new dataset contains no missing coordinates and
has at most $\varepsilon N$ corrupted examples.
Moreover, \Cref{claim:composition} shows that the new dataset can be regarded as
an $\frac{\varepsilon N}{N'}\approx \frac{\varepsilon}{\gamma}$-corrupted
$1$-complete dataset.
This allows us to apply existing mean estimation algorithms to the new dataset.
\begin{figure*}[h]
\centering
\begin{subfigure}[t]{.32\textwidth}
\begin{equation*}
\begin{blockarray}{c@{\hspace{7pt}}cccc}
\mu & 1.0 & 2.0 & 3.0 & 4.0 \\
\begin{block}{c@{\hspace{7pt}}[cccc]}
X\sps 1 & \org{1.2} & \org{1.8} & 2.9 & 4.0 
\\
X\sps 2 & 0.9 & 2.2 & \org{2.8} & \org{3.9} 
\\ \cmidrule(r){2-5} 
X\sps 3 & \org{0.8} & 1.9 & \org{3.1} & 4.1 
\\
X\sps 4 & 1.1 & \org{2.1} & 2.9 & \org{4.1} 
\\ \cmidrule(r){2-5} 
X\sps 5 & \org{1.0} & \org{2.0} & 3.0 & 3.8 
\\
X\sps 6 & 1.2 & 2.0 & \org{2.9} & 4.2 
\\
X\sps 7 & 1.2 & 2.1 & 3.2 & \org{3.9} 
\\
\end{block}
\end{blockarray}
\end{equation*}
\end{subfigure}
\begin{subfigure}[t]{.32\textwidth}
\begin{equation*}
\begin{blockarray}{c@{\hspace{7pt}}cccc}
\mu & 1.0 & 2.0 & 3.0 & 4.0 \\
\begin{block}{c@{\hspace{7pt}}[cccc]}
X\sps 1 & 1.2 & 1.8 & * & *
\\
X\sps 2 & * & * & 2.8 & 3.9
\\ \cmidrule(r){2-5} 
X\sps 3 & 0.8 & * & 3.1 & *
\\
X\sps 4 & * & \mathbf{2.6} & * & \mathbf{5.0}
\\ \cmidrule(r){2-5} 
X\sps 5 & 1.0 & 2.0 & * & *
\\
X\sps 6 & * & * & 2.9 & * 
\\
X\sps 7 & * & * & * & 3.9 
\\
\end{block}
\end{blockarray}
\end{equation*}
\end{subfigure}
\begin{subfigure}[t]{.32\textwidth}
\begin{equation*}
\begin{blockarray}{c@{\hspace{7pt}}cccc}
\mu & 1.0 & 2.0 & 3.0 & 4.0 \\
\begin{block}{c@{\hspace{7pt}}[cccc]}
\widetilde X\sps 1 & 1.2 & 1.8 & 2.8 & 3.9 
\\
\widetilde X\sps 2 & 0.8 & \mathbf{2.6} & 3.1 & \mathbf{5.0} 
\\
\widetilde X\sps 3 & 1.0 & 2.0 & 2.9 & 3.9 
\\
\end{block}
\end{blockarray}
\end{equation*}
\end{subfigure}
\caption{If we choose the present entries as 
the orange ones in the left matrix,
the three groups of the observed examples (the middle matrix) 
correspond to 
three new examples with no missing coordinates
(the right matrix). Here, $X\sps 4$ is corrupted, and so is 
$\widetilde X\sps 2$.
}
\label{fig:partition}
\end{figure*}

To extend this reduction to cases where the present entries are chosen by the adversary,
we need to figure out how to partition the original examples into groups.
The $N'$ groups of the original examples 
essentially correspond to a hash function
$h:\{1,\ldots,N\}\rightarrow \{1,\ldots,N'\}$
mapping original examples to new examples.
Each new example $\widetilde X\sps {i'}$ is defined by 
$\widetilde X\sps {i'}_j = X\sps i_j$ for some
$i\in h^{-1}(i')$ satisfying $X\sps i_j\neq *$.
For definiteness, let us always choose the smallest one
if there are multiple such $i$,
and define $\widetilde X\sps{i'}_j = *$
if no such $i$ exists. 
For example, if $h$ maps $1,\ldots,7$ into $2,1,3,1,1,3,2$, the dataset given in \Cref{fig:data} is transformed to
\begin{equation*}
\begin{bmatrix}
0.9 & \mathbf{2.6} & 2.8 & 3.9 \\
1.2 & 2.1 & 2.9 & *   \\
* & 2.0 & 2.9 & 4.1
\end{bmatrix}.
\end{equation*}
Unfortunately, even if every coordinate $j$ is present in $\gamma N$ examples
chosen uniformly at random and independently for every $j$,
one has to choose $N'$ to be as small as $o(\gamma N)$ when $d,N\rightarrow +\infty$ 
in order to effectively eliminate the missing entries by hashing
(see \Cref{claim:hashing-lower} in \Cref{sec:proof-hashing} for details).
In that case, the total number of present entries will reduce significantly
in the new dataset.
Moreover, if $\gamma$ is only a constant factor larger than $\varepsilon$,
the total number of new examples may become smaller than 
the number of corrupted original examples ($N' \leq o(\gamma N) < \varepsilon N$),
and the majority of the new examples will likely have corrupted coordinates.

\subsection{Pre-processing via Random Hashing}
\label{sec:preprocessing}
We show that, for a sufficiently small constant $c > 0$,
we can assume without loss of generality that $\gamma\geq c$ in our main theorems (\Cref{thm:Gaussian,thm:2nd-moment}) 
by performing a pre-processing step.
By ``sufficiently small'', we mean there is a large integer $B>2$ such that
$c \leq (B - 2)/B^2$.
If $\gamma < c$ instead, we hash the $N$ examples into $N':= B\lceil \gamma N\rceil$ new examples as in \Cref{sec:hashing}. 
We use a random hash function $h$
where $h(i)$ is independently uniformly chosen from $\{1,\ldots,N'\}$ 
for every $i = 1,\ldots,N$.
There are at most $\varepsilon N \leq \frac{\varepsilon}{B\gamma}N'$
corrupted new examples, so
\Cref{claim:composition} guarantees that the new dataset is an $\varepsilon'$-corrupted dataset with the same $\mu$ for $\varepsilon' :=  \frac{\varepsilon}{B\gamma}$. 
To lower-bound the number of present coordinates in the new dataset, 
we note that $\widetilde X\sps {i'}_j$ is present
as long as any of the original examples 
$i\in\{1,\ldots,N\}$ with $X\sps{i}_j\neq *$ is hashed to $i'$.
When the original dataset is $\gamma$-complete,
the number of new examples with coordinate $j$ present
is at least the number of non-empty bins if we throw $\lceil \gamma N\rceil$
balls into $B\lceil \gamma N\rceil$ bins uniformly at random.
If we throw the balls one by one,
the conditional probability of each ball being thrown into an empty bin is
at least $(B - 1)/B$, so the number of non-empty bins is
at least $((B -2)/B)\lceil \gamma N\rceil$ with probability 
$1 - \delta/2d$ by the Chernoff bound (recall our assumption $\gamma N \ge \Omega(\log(d/\delta))$).
By a union bound over all the $d$ coordinates,
the new dataset is $\gamma' := ((B - 2)/B)\lceil \gamma N\rceil/N' = (B- 2)/B^2$-complete with probability at least $1 - \delta/2$
conditioned on the original dataset being $\gamma$-complete.

We can now run the algorithms in our main theorems on the newly created $\varepsilon'$-corrupted $\gamma'$-complete dataset of size $N'$ where, indeed, $\gamma' = (B - 2)/B^2 \geq c$. Note that the values $\gamma/\varepsilon$ and $\gamma N$ barely change: $\gamma'/\varepsilon' = ((B - 2)/B) \gamma/\varepsilon$ and $\gamma'N' \geq ((B - 2)/B) \gamma N$. Therefore, all the assumptions of our main theorems still hold on the new dataset and our guarantees can translate back to the original dataset, up to tiny loss in the constants. 
The pre-processing step clearly runs in $O(Nd)$ time, and the size of the dataset does not increase.
\section{Conclusions and Future Research}
Motivated by applications such as crowdsourcing,
we asked the natural theoretical question  
of robust mean estimation on incomplete data
and solved it (nearly) optimally
in terms of the error guarantee, the sample complexity,
and the
running time.
As we discussed
in \Cref{sec:results},
three natural questions remain unanswered
and would be interesting topics for future research:
1) achieving sub-Gaussian rates in \Cref{thm:2nd-moment},
2) extending \Cref{thm:Gaussian} to general known
covariances,
and
3) robustly estimating the covariance 
given incomplete data with adversarial corruption.
Another interesting direction is to generalize our algorithms to settings where a majority of the present entries may be corrupted, i.e., $\gamma < 2\varepsilon$. This has been considered in the complete-data setting by \citet{charikar2017learning} and \citet{cherapanamjeri2020list} using list-decoding, where the algorithm outputs a list of estimates with one of them likely being accurate. The list-decoding idea also applies when the underlying distribution is a mixture, where we want to estimate the mean of each component. Extending this list-decoding idea to incomplete data would lead to practical algorithms that can understand the views of diverse subgroups in the population from crowdsourced data.
We hope that our iterative guessing-and-improving algorithm
provides insights for answering these questions
and for solving other problems
in the incomplete-data setting.
\section*{Acknowledgments}
We thank Michael P.\ Kim for valuable discussions at early stages of this work.
We thank Jacob Steinhardt and Banghua Zhu for discussions on related research.
We thank Hilal Asi, Mayee Chen, Dan Fu, Vatsal Sharan, and anonymous reviewers 
for helpful comments on earlier versions of this paper.
\bibliographystyle{plainnat}
\bibliography{ref}

\begin{thebibliography}{62}
\providecommand{\natexlab}[1]{#1}
\providecommand{\url}[1]{\texttt{#1}}
\expandafter\ifx\csname urlstyle\endcsname\relax
  \providecommand{\doi}[1]{doi: #1}\else
  \providecommand{\doi}{doi: \begingroup \urlstyle{rm}\Url}\fi

\bibitem[Balakrishnan et~al.(2017)Balakrishnan, Du, Li, and
  Singh]{balakrishnan2017computationally}
Sivaraman Balakrishnan, Simon~S Du, Jerry Li, and Aarti Singh.
\newblock Computationally efficient robust sparse estimation in high
  dimensions.
\newblock In \emph{Conference on Learning Theory}, pages 169--212, 2017.

\bibitem[Batu et~al.(2004)Batu, Kannan, Khanna, and
  McGregor]{batu2004reconstructing}
Tu\v{g}kan Batu, Sampath Kannan, Sanjeev Khanna, and Andrew McGregor.
\newblock Reconstructing strings from random traces.
\newblock In \emph{Proceedings of the {F}ifteenth {A}nnual {ACM}-{SIAM}
  {S}ymposium on {D}iscrete {A}lgorithms}, pages 910--918. ACM, New York, 2004.

\bibitem[Bhatia et~al.(2015)Bhatia, Jain, and Kar]{bhatia2015robust}
Kush Bhatia, Prateek Jain, and Purushottam Kar.
\newblock Robust regression via hard thresholding.
\newblock In \emph{Advances in Neural Information Processing Systems}, pages
  721--729, 2015.

\bibitem[Bhatia et~al.(2017)Bhatia, Jain, Kamalaruban, and
  Kar]{bhatia2017consistent}
Kush Bhatia, Prateek Jain, Parameswaran Kamalaruban, and Purushottam Kar.
\newblock Consistent robust regression.
\newblock In \emph{Advances in Neural Information Processing Systems}, pages
  2110--2119, 2017.

\bibitem[Candes and Plan(2010)]{candes2010matrix}
Emmanuel~J Candes and Yaniv Plan.
\newblock Matrix completion with noise.
\newblock \emph{Proceedings of the IEEE}, 98\penalty0 (6):\penalty0 925--936,
  2010.

\bibitem[Cand\`es and Recht(2009)]{candes2009exact}
Emmanuel~J. Cand\`es and Benjamin Recht.
\newblock Exact matrix completion via convex optimization.
\newblock \emph{Found. Comput. Math.}, 9\penalty0 (6):\penalty0 717--772, 2009.
\newblock ISSN 1615-3375.
\newblock \doi{10.1007/s10208-009-9045-5}.
\newblock URL \url{https://doi.org/10.1007/s10208-009-9045-5}.

\bibitem[Cand\`es and Tao(2010)]{candes2010power}
Emmanuel~J. Cand\`es and Terence Tao.
\newblock The power of convex relaxation: near-optimal matrix completion.
\newblock \emph{IEEE Trans. Inform. Theory}, 56\penalty0 (5):\penalty0
  2053--2080, 2010.
\newblock ISSN 0018-9448.
\newblock \doi{10.1109/TIT.2010.2044061}.
\newblock URL \url{https://doi.org/10.1109/TIT.2010.2044061}.

\bibitem[Cand{\`e}s et~al.(2011)Cand{\`e}s, Li, Ma, and
  Wright]{candes2011robust}
Emmanuel~J Cand{\`e}s, Xiaodong Li, Yi~Ma, and John Wright.
\newblock Robust principal component analysis?
\newblock \emph{Journal of the ACM (JACM)}, 58\penalty0 (3):\penalty0 1--37,
  2011.

\bibitem[Chandrasekaran et~al.(2011)Chandrasekaran, Sanghavi, Parrilo, and
  Willsky]{chandrasekaran2011rank}
Venkat Chandrasekaran, Sujay Sanghavi, Pablo~A Parrilo, and Alan~S Willsky.
\newblock Rank-sparsity incoherence for matrix decomposition.
\newblock \emph{SIAM Journal on Optimization}, 21\penalty0 (2):\penalty0
  572--596, 2011.

\bibitem[Charikar et~al.(2017)Charikar, Steinhardt, and
  Valiant]{charikar2017learning}
Moses Charikar, Jacob Steinhardt, and Gregory Valiant.
\newblock Learning from untrusted data.
\newblock In \emph{Proceedings of the 49th Annual ACM SIGACT Symposium on
  Theory of Computing}, pages 47--60, 2017.

\bibitem[Chen et~al.(2018)Chen, Gao, Ren, et~al.]{chen2018robust}
Mengjie Chen, Chao Gao, Zhao Ren, et~al.
\newblock Robust covariance and scatter matrix estimation under {H}uber’s
  contamination model.
\newblock \emph{The Annals of Statistics}, 46\penalty0 (5):\penalty0
  1932--1960, 2018.

\bibitem[Chen et~al.(2020)Chen, Li, and Moitra]{chen2020learning}
Sitan Chen, Jerry Li, and Ankur Moitra.
\newblock Efficiently learning structured distributions from untrusted batches.
\newblock In \emph{Proceedings of the 52nd Annual ACM SIGACT Symposium on
  Theory of Computing}, pages 960--973, 2020.

\bibitem[Cheng et~al.(2019)Cheng, Diakonikolas, and Ge]{cheng2019high}
Yu~Cheng, Ilias Diakonikolas, and Rong Ge.
\newblock High-dimensional robust mean estimation in nearly-linear time.
\newblock In \emph{Proceedings of the Thirtieth Annual ACM-SIAM Symposium on
  Discrete Algorithms}, pages 2755--2771. SIAM, 2019.

\bibitem[Cheng et~al.(2020)Cheng, Diakonikolas, Ge, and
  Soltanolkotabi]{cheng2020high}
Yu~Cheng, Ilias Diakonikolas, Rong Ge, and Mahdi Soltanolkotabi.
\newblock High-dimensional robust mean estimation via gradient descent.
\newblock In \emph{International Conference on Machine Learning}, 2020.

\bibitem[{Cherapanamjeri} et~al.(2020){Cherapanamjeri}, {Mohanty}, and
  {Yau}]{cherapanamjeri2020list}
Y.~{Cherapanamjeri}, S.~{Mohanty}, and M.~{Yau}.
\newblock List decodable mean estimation in nearly linear time.
\newblock In \emph{2020 IEEE 61st Annual Symposium on Foundations of Computer
  Science (FOCS)}, pages 141--148, 2020.
\newblock \doi{10.1109/FOCS46700.2020.00022}.

\bibitem[Dempster et~al.(1977)Dempster, Laird, and Rubin]{dempster1977maximum}
A.~P. Dempster, N.~M. Laird, and D.~B. Rubin.
\newblock Maximum likelihood from incomplete data via the {EM} algorithm.
\newblock \emph{J. Roy. Statist. Soc. Ser. B}, 39\penalty0 (1):\penalty0 1--38,
  1977.
\newblock ISSN 0035-9246.
\newblock URL
  \url{http://links.jstor.org/sici?sici=0035-9246(1977)39:1<1:MLFIDV>2.0.CO;2-Z&origin=MSN}.
\newblock With discussion.

\bibitem[Depersin and Lecu{\'e}(2019)]{depersin2019robust}
Jules Depersin and Guillaume Lecu{\'e}.
\newblock Robust subgaussian estimation of a mean vector in nearly linear time.
\newblock \emph{arXiv preprint arXiv:1906.03058}, 2019.

\bibitem[Diakonikolas and Kane(2019)]{diakonikolas2019recent}
Ilias Diakonikolas and Daniel~M Kane.
\newblock Recent advances in algorithmic high-dimensional robust statistics.
\newblock \emph{arXiv preprint arXiv:1911.05911}, 2019.

\bibitem[Diakonikolas et~al.(2016)Diakonikolas, Kane, and
  Stewart]{diakonikolas2016robust}
Ilias Diakonikolas, Daniel~M Kane, and Alistair Stewart.
\newblock Robust learning of fixed-structure bayesian networks.
\newblock \emph{CoRR, abs/1606.07384}, 2016.

\bibitem[Diakonikolas et~al.(2017{\natexlab{a}})Diakonikolas, Kamath, Kane, Li,
  Moitra, and Stewart]{diakonikolas2017being}
Ilias Diakonikolas, Gautam Kamath, Daniel~M Kane, Jerry Li, Ankur Moitra, and
  Alistair Stewart.
\newblock Being robust (in high dimensions) can be practical.
\newblock In \emph{Proceedings of the 34th International Conference on Machine
  Learning-Volume 70}, pages 999--1008. JMLR. org, 2017{\natexlab{a}}.

\bibitem[Diakonikolas et~al.(2017{\natexlab{b}})Diakonikolas, Kamath, Kane, Li,
  Moitra, and Stewart]{diakonikolas2017robustly-arxiv}
Ilias Diakonikolas, Gautam Kamath, Daniel~M Kane, Jerry Li, Ankur Moitra, and
  Alistair Stewart.
\newblock Robustly learning a gaussian: Getting optimal error, efficiently.
\newblock \emph{arXiv preprint arXiv:1704.03866}, 2017{\natexlab{b}}.

\bibitem[Diakonikolas et~al.(2017{\natexlab{c}})Diakonikolas, Kane, and
  Stewart]{diakonikolas2017statistical}
Ilias Diakonikolas, Daniel~M Kane, and Alistair Stewart.
\newblock Statistical query lower bounds for robust estimation of
  high-dimensional {Gaussians} and {Gaussian} mixtures.
\newblock In \emph{2017 IEEE 58th Annual Symposium on Foundations of Computer
  Science (FOCS)}, pages 73--84. IEEE, 2017{\natexlab{c}}.

\bibitem[Diakonikolas et~al.(2018)Diakonikolas, Kamath, Kane, Li, Moitra, and
  Stewart]{diakonikolas2018robustly}
Ilias Diakonikolas, Gautam Kamath, Daniel~M Kane, Jerry Li, Ankur Moitra, and
  Alistair Stewart.
\newblock Robustly learning a {Gaussian}: Getting optimal error, efficiently.
\newblock In \emph{Proceedings of the Twenty-Ninth Annual ACM-SIAM Symposium on
  Discrete Algorithms}, pages 2683--2702. SIAM, 2018.

\bibitem[Diakonikolas et~al.(2019{\natexlab{a}})Diakonikolas, Kamath, Kane, Li,
  Moitra, and Stewart]{diakonikolas2019robust}
Ilias Diakonikolas, Gautam Kamath, Daniel Kane, Jerry Li, Ankur Moitra, and
  Alistair Stewart.
\newblock Robust estimators in high-dimensions without the computational
  intractability.
\newblock \emph{SIAM Journal on Computing}, 48\penalty0 (2):\penalty0 742--864,
  2019{\natexlab{a}}.

\bibitem[Diakonikolas et~al.(2019{\natexlab{b}})Diakonikolas, Kong, and
  Stewart]{diakonikolas2019efficient}
Ilias Diakonikolas, Weihao Kong, and Alistair Stewart.
\newblock Efficient algorithms and lower bounds for robust linear regression.
\newblock In \emph{Proceedings of the Thirtieth Annual ACM-SIAM Symposium on
  Discrete Algorithms}, pages 2745--2754. SIAM, 2019{\natexlab{b}}.

\bibitem[Diakonikolas et~al.(2020)Diakonikolas, Kane, and
  Pensia]{diakonikolas2020outlier}
Ilias Diakonikolas, Daniel~M Kane, and Ankit Pensia.
\newblock Outlier robust mean estimation with subgaussian rates via stability.
\newblock \emph{Advances in Neural Information Processing Systems}, 33, 2020.

\bibitem[Dong et~al.(2019{\natexlab{a}})Dong, Hopkins, and Li]{dong2019quantum}
Yihe Dong, Samuel Hopkins, and Jerry Li.
\newblock Quantum entropy scoring for fast robust mean estimation and improved
  outlier detection.
\newblock In \emph{Advances in Neural Information Processing Systems}, pages
  6065--6075, 2019{\natexlab{a}}.

\bibitem[Dong et~al.(2019{\natexlab{b}})Dong, Hopkins, and
  Li]{dong2019quantum-arxiv}
Yihe Dong, Samuel~B Hopkins, and Jerry Li.
\newblock Quantum entropy scoring for fast robust mean estimation and improved
  outlier detection.
\newblock \emph{arXiv preprint arXiv:1906.11366}, 2019{\natexlab{b}}.

\bibitem[Gao et~al.(2018)Gao, Xie, Xie, and Xu]{gao2018robust}
Rui Gao, Liyan Xie, Yao Xie, and Huan Xu.
\newblock Robust hypothesis testing using {Wasserstein} uncertainty sets.
\newblock In \emph{Advances in Neural Information Processing Systems}, pages
  7902--7912, 2018.

\bibitem[G{\"u}l(2017)]{gul2017robust}
G{\"o}khan G{\"u}l.
\newblock \emph{Robust and distributed hypothesis testing}.
\newblock Springer, 2017.

\bibitem[G{\"u}l and Zoubir(2017)]{gul2017minimax}
G{\"o}khan G{\"u}l and Abdelhak~M Zoubir.
\newblock Minimax robust hypothesis testing.
\newblock \emph{IEEE Transactions on Information Theory}, 63\penalty0
  (9):\penalty0 5572--5587, 2017.

\bibitem[Holenstein et~al.(2008)Holenstein, Mitzenmacher, Panigrahy, and
  Wieder]{holenstein2008trace}
Thomas Holenstein, Michael Mitzenmacher, Rina Panigrahy, and Udi Wieder.
\newblock Trace reconstruction with constant deletion probability and related
  results.
\newblock In \emph{Proceedings of the {N}ineteenth {A}nnual {ACM}-{SIAM}
  {S}ymposium on {D}iscrete {A}lgorithms}, pages 389--398. ACM, New York, 2008.

\bibitem[Hopkins et~al.(2020)Hopkins, Li, and Zhang]{hopkins2020robust}
Samuel~B Hopkins, Jerry Li, and Fred Zhang.
\newblock Robust and heavy-tailed mean estimation made simple, via regret
  minimization.
\newblock \emph{Advances in Neural Information Processing Systems}, 33, 2020.

\bibitem[Huber and Strassen(1973)]{huber1973minimax}
Peter~J Huber and Volker Strassen.
\newblock Minimax tests and the {Neyman}-{Pearson} lemma for capacities.
\newblock \emph{The Annals of Statistics}, pages 251--263, 1973.

\bibitem[Huber et~al.(1964)]{huber1964robust}
Peter~J Huber et~al.
\newblock Robust estimation of a location parameter.
\newblock \emph{The Annals of Mathematical Statistics}, 35\penalty0
  (1):\penalty0 73--101, 1964.

\bibitem[Johnson and Preparata(1978)]{johnson1978densest}
David~S. Johnson and Franco~P Preparata.
\newblock The densest hemisphere problem.
\newblock \emph{Theoretical Computer Science}, 6\penalty0 (1):\penalty0
  93--107, 1978.

\bibitem[Kulkarni et~al.(2013)Kulkarni, Wei, Le, Chia, Papadopoulos, Cheng,
  Koller, and Klemmer]{kulkarni2013peer}
Chinmay Kulkarni, Koh~Pang Wei, Huy Le, Daniel Chia, Kathryn Papadopoulos,
  Justin Cheng, Daphne Koller, and Scott~R Klemmer.
\newblock Peer and self assessment in massive online classes.
\newblock \emph{ACM Transactions on Computer-Human Interaction (TOCHI)},
  20\penalty0 (6):\penalty0 1--31, 2013.

\bibitem[Lai et~al.(2016)Lai, Rao, and Vempala]{lai2016agnostic}
Kevin~A Lai, Anup~B Rao, and Santosh Vempala.
\newblock Agnostic estimation of mean and covariance.
\newblock In \emph{2016 IEEE 57th Annual Symposium on Foundations of Computer
  Science (FOCS)}, pages 665--674. IEEE, 2016.

\bibitem[Lei et~al.(2020)Lei, Luh, Venkat, and Zhang]{lei2020fast}
Zhixian Lei, Kyle Luh, Prayaag Venkat, and Fred Zhang.
\newblock A fast spectral algorithm for mean estimation with sub-gaussian
  rates.
\newblock In \emph{Conference on Learning Theory}, pages 2598--2612. PMLR,
  2020.

\bibitem[Levy(2008)]{levy2008robust}
Bernard~C Levy.
\newblock Robust hypothesis testing with a relative entropy tolerance.
\newblock \emph{IEEE Transactions on Information Theory}, 55\penalty0
  (1):\penalty0 413--421, 2008.

\bibitem[Li(2018)]{li2018principled}
Jerry~Zheng Li.
\newblock \emph{Principled approaches to robust machine learning and beyond}.
\newblock PhD thesis, Massachusetts Institute of Technology, 2018.

\bibitem[Liu et~al.(2019)Liu, Li, and Caramanis]{liu2019high}
Liu Liu, Tianyang Li, and Constantine Caramanis.
\newblock High dimensional robust $ m $-estimation: Arbitrary corruption and
  heavy tails.
\newblock \emph{arXiv preprint arXiv:1901.08237}, 2019.

\bibitem[Liu et~al.(2020{\natexlab{a}})Liu, Shen, Li, and
  Caramanis]{liu2020high}
Liu Liu, Yanyao Shen, Tianyang Li, and Constantine Caramanis.
\newblock High dimensional robust sparse regression.
\newblock In \emph{International Conference on Artificial Intelligence and
  Statistics}, pages 411--421. PMLR, 2020{\natexlab{a}}.

\bibitem[Liu et~al.(2020{\natexlab{b}})Liu, Park, Palumbo, Rekatsinas, and
  Tzamos]{liu2020robust}
Zifan Liu, Jongho Park, Nils Palumbo, Theodoros Rekatsinas, and Christos
  Tzamos.
\newblock Robust mean estimation under coordinate-level corruption with missing
  entries.
\newblock \emph{arXiv preprint arXiv:2002.04137}, 2020{\natexlab{b}}.

\bibitem[Lugosi and Mendelson(2019)]{lugosi2019mean}
G{\'a}bor Lugosi and Shahar Mendelson.
\newblock Mean estimation and regression under heavy-tailed distributions: A
  survey.
\newblock \emph{Foundations of Computational Mathematics}, 19\penalty0
  (5):\penalty0 1145--1190, 2019.

\bibitem[Lugosi et~al.(2019)Lugosi, Mendelson, et~al.]{lugosi2019sub}
G{\'a}bor Lugosi, Shahar Mendelson, et~al.
\newblock Sub-gaussian estimators of the mean of a random vector.
\newblock \emph{Annals of Statistics}, 47\penalty0 (2):\penalty0 783--794,
  2019.

\bibitem[Maurer et~al.(2019)]{maurer2019bernstein}
Andreas Maurer et~al.
\newblock A {B}ernstein-type inequality for functions of bounded interaction.
\newblock \emph{Bernoulli}, 25\penalty0 (2):\penalty0 1451--1471, 2019.

\bibitem[Piech et~al.(2013)Piech, Huang, Chen, Do, Ng, and
  Koller]{piech2013tuned}
Chris Piech, Jonathan Huang, Zhenghao Chen, Chuong Do, Andrew Ng, and Daphne
  Koller.
\newblock Tuned models of peer assessment in {MOOCs}.
\newblock \emph{arXiv preprint arXiv:1307.2579}, 2013.

\bibitem[Qiao and Valiant(2018)]{qiao2018learning}
Mingda Qiao and Gregory Valiant.
\newblock Learning discrete distributions from untrusted batches.
\newblock In \emph{9th Innovations in Theoretical Computer Science Conference
  (ITCS 2018)}. Schloss Dagstuhl-Leibniz-Zentrum fuer Informatik, 2018.

\bibitem[Rubin(1979)]{rubin1979illustrating}
D.~Rubin.
\newblock Illustrating the use of multiple imputations to handle nonresponse in
  sample surveys.
\newblock In \emph{Proceedings of the 42nd session of the {I}nternational
  {S}tatistical {I}nstitute, {V}ol. 2 ({M}anila, 1979)}, volume~48, pages
  517--532, 1979.

\bibitem[Rubin(1976)]{rubin1976inference}
Donald~B. Rubin.
\newblock Inference and missing data.
\newblock \emph{Biometrika}, 63\penalty0 (3):\penalty0 581--592, 1976.
\newblock ISSN 0006-3444.
\newblock \doi{10.1093/biomet/63.3.581}.
\newblock URL \url{https://doi.org/10.1093/biomet/63.3.581}.
\newblock With comments by R. J. A. Little and a reply by the author.

\bibitem[Steinhardt et~al.(2016)Steinhardt, Valiant, and
  Charikar]{steinhardt2016avoiding}
Jacob Steinhardt, Gregory Valiant, and Moses Charikar.
\newblock Avoiding imposters and delinquents: Adversarial crowdsourcing and
  peer prediction.
\newblock In \emph{Advances in Neural Information Processing Systems}, pages
  4439--4447, 2016.

\bibitem[Steinhardt et~al.(2018)Steinhardt, Charikar, and
  Valiant]{steinhardt2018resilience}
Jacob Steinhardt, Moses Charikar, and Gregory Valiant.
\newblock Resilience: A criterion for learning in the presence of arbitrary
  outliers.
\newblock In \emph{9th Innovations in Theoretical Computer Science Conference
  (ITCS 2018)}. Schloss Dagstuhl-Leibniz-Zentrum fuer Informatik, 2018.

\bibitem[Tropp(2015)]{tropp2015introduction}
Joel~A Tropp.
\newblock An introduction to matrix concentration inequalities.
\newblock \emph{Foundations and Trends{\textregistered} in Machine Learning},
  8\penalty0 (1-2):\penalty0 1--230, 2015.

\bibitem[Tukey(1960)]{tukey1960survey}
John~W Tukey.
\newblock A survey of sampling from contaminated distributions.
\newblock \emph{Contributions to probability and statistics}, pages 448--485,
  1960.

\bibitem[Tukey(1975)]{tukey1975mathematics}
John~W Tukey.
\newblock Mathematics and the picturing of data.
\newblock In \emph{Proceedings of the International Congress of Mathematicians,
  Vancouver, 1975}, volume~2, pages 523--531, 1975.

\bibitem[Verdu and Poor(1984)]{verdu1984minimax}
Sergio Verdu and H~Poor.
\newblock On minimax robustness: A general approach and applications.
\newblock \emph{IEEE Transactions on Information Theory}, 30\penalty0
  (2):\penalty0 328--340, 1984.

\bibitem[Vuurens et~al.(2011)Vuurens, de~Vries, and Eickhoff]{vuurens2011much}
Jeroen Vuurens, Arjen~P de~Vries, and Carsten Eickhoff.
\newblock How much spam can you take? {A}n analysis of crowdsourcing results to
  increase accuracy.
\newblock In \emph{Proc. ACM SIGIR Workshop on Crowdsourcing for Information
  Retrieval (CIR’11)}, pages 21--26, 2011.

\bibitem[Xu et~al.(2010)Xu, Caramanis, and Mannor]{xu2010principal}
Huan Xu, Constantine Caramanis, and Shie Mannor.
\newblock Principal component analysis with contaminated data: {T}he high
  dimensional case.
\newblock \emph{arXiv preprint arXiv:1002.4658}, 2010.

\bibitem[Zhu et~al.(2019)Zhu, Jiao, and Steinhardt]{zhu2019generalized}
Banghua Zhu, Jiantao Jiao, and Jacob Steinhardt.
\newblock Generalized resilience and robust statistics.
\newblock \emph{arXiv preprint arXiv:1909.08755}, 2019.

\bibitem[Zhu et~al.(2020{\natexlab{a}})Zhu, Jiao, and Steinhardt]{zhu2020does}
Banghua Zhu, Jiantao Jiao, and Jacob Steinhardt.
\newblock When does the {Tukey} median work?
\newblock In \emph{2020 IEEE International Symposium on Information Theory
  (ISIT)}, pages 1201--1206. IEEE, 2020{\natexlab{a}}.

\bibitem[Zhu et~al.(2020{\natexlab{b}})Zhu, Jiao, and
  Steinhardt]{zhu2020robust}
Banghua Zhu, Jiantao Jiao, and Jacob Steinhardt.
\newblock Robust estimation via generalized quasi-gradients.
\newblock \emph{arXiv preprint arXiv:2005.14073}, 2020{\natexlab{b}}.

\end{thebibliography}
\newpage
\appendix
\section{Proof of Claim \ref{claim:composition}}
\label{sec:proof-composition}
\begin{proof}
It is quite clear that $\cD$ has mean $\mu$ by looking at the mean of each coordinate. We now show that $\cD\in\cP$. 

When $\cP = \cP_1(\eta)$, every $\cD\sps i$ has covariance $\eta^2 I$, that is, the variance of every coordinate is $\eta^2$ and the covariances between coordinates are zero. This implies that $\cD$ also has covariance $\eta I$. We now show that $\cD$ is $1$-sub-Gaussian. For any vector $v\in\bR^d$, it can be decomposed as 
\begin{equation}
\label{eq:decompose-v}
v = v\sps 1+\cdots+v\sps m, 
\end{equation}
where $v\sps i_j$ is zero whenever $j\notin J\sps i$. Now we have
\begin{align*}
\bE_{X\sim\cD}[\exp((X - \mu)^\top v)] & = \prod_{i=1}^m\bE_{X\sps i\sim\cD\sps i}[\exp((X\sps i - \mu)^\top v\sps i)]\tag{$X\sps i$'s are independent}\\
& \leq \prod_{i=1}^m\exp(\|v\sps i\|_2^2/2)\tag{$\cD\sps i$'s are $1$-sub-Gaussian}\\
& = \exp(\|v\|_2^2/2). \tag{$v\sps i$'s are orthogonal}
\end{align*}
This proves that $\cD$ is $1$-sub-Gaussian and thus $\cD\in\cP_1(\eta) = \cP$. When $\cP = \cP_2$, we note that a distribution has covariance $\Sigma \preceq I$ if and only if $\bE_{X\sim\cD}[((X - \mu)^\top v)^2] \leq \|v\|_2^2$ for all vector $v\in \bR^d$. Using \eqref{eq:decompose-v} again, the following calculations prove $\cD\in\cP_2 = \cP$:
\begin{align*}
\bE_{X\sim\cD}[((X - \mu)^\top v)^2] & = \sum_{i=1}^m\bE_{X\sps i\sim\cD\sps i}[((X\sps i - \mu)^\top v\sps i)^2] \tag{$X\sps i$'s are independent}\\
& \leq \sum_{i=1}^m \|v\sps i\|_2^2 \tag{$\cD\sps i$'s have covariances $\Sigma\preceq I$}\\
& = \|v\|_2^2. \tag{$v\sps i$'s are orthogonal}
\end{align*}
\end{proof}
\section{Proof of Theorem \ref{thm:no-missing}}
\label{sec:QUE}
\begin{proof}
We prove the theorem by reducing it to 
Theorem 3.5 and Theorem 4.7 of \cite{dong2019quantum-arxiv}. 

We first show that we can assume without loss of generality
that $\varepsilon \geq 1/N,\eta \in\{0,1\}$ and $\beta \leq 100$.
When $\varepsilon < 1/N$, we have $S = \{1,\ldots,N\}$ in \Cref{def:no-missing}
and the theorem becomes trivial.
When $\eta = 0$, we can scale the whole dataset and $\beta$
by the same factor without changing the problem,
so we can assume $\beta \leq 100$ in this case.
When $\eta > 0$, we can scale the whole dataset, $\eta$, and $\beta$
by the same factor without changing the problem,
so we can assume $\eta = 1$.
If $\beta \geq 100$ after setting $\eta = 1$, 
we note that any $(\varepsilon,1,\beta)$-good dataset
is also an $\Big(\varepsilon,0,\sqrt{\beta^2 + 1}\Big)$-good dataset,
so the problem reduces to the $\eta = 0$ case 
with only minor loss in the constants.

Next, we show that the dataset can be further processed to satisfy the following:
\begin{align}
& \max_{i\in G}\|X\sps i\|  \leq O\Big(\sqrt N\Big), \ \textnormal{and} \label{eq:QUE-assumption1}\\
& \Big\|\frac 1N\sum_{i=1}^N 
(X\sps i - \hat\mu)(X\sps i - \hat\mu)^\top
 - \eta^2 I\Big\|_\op  \leq O(N\varepsilon),\label{eq:QUE-assumption2}
\end{align}
where $\hat \mu$ is the empirical average of $X\sps 1,\ldots,X\sps N$.
The third property of $(\varepsilon,\eta,\beta)$-goodness implies that
\begin{equation*}
\frac 1{|G|}\sum_{i\in G}(X\sps i - \mu)(X\sps i - \mu)^\top 
\preceq (\beta^2 + \eta^2)I \preceq O(1) I.
\end{equation*}
Therefore,
\begin{equation*}
\forall i = 1,\ldots,N, \quad (X\sps i - \mu)(X\sps i - \mu)^\top \preceq O(N) I,
\end{equation*}
or in other words, $\max_{i\in G}\|X\sps i - \mu\| \leq O(\sqrt N)$.
We can then use the naive-pruning lemma in \citet[Lemma 2.3]{dong2019quantum-arxiv} to re-center the data points\footnote{
Examples outside $G$ can be changed arbitrarily
without changing the goodness property.
} 
so that \eqref{eq:QUE-assumption1} is satisfied.

Now we prove \eqref{eq:QUE-assumption2}.
We first observe that
\begin{align*}
\|\hat\mu - \mu\|_2 
& \leq \frac{|G|}{N}\Big\|\frac {1}{|G|}\sum_{i\in G}(X\sps i - \mu)\Big\|_2
+ \frac{N - |G|}{N}\Big\|\frac{1}{N - |G|}\sum_{i\notin G}(X\sps i - \mu)\Big\|_2\\
& \leq \frac{|G|}N\cdot \beta\sqrt\varepsilon +  \frac{N - |G|}{N}\cdot O\Big(\sqrt N\Big)\\
& \leq O\Big(\varepsilon \sqrt N\Big).\tag{because $\beta \leq O(1)$ and $\varepsilon \geq 1/N$}
\end{align*} 
Now we have
\begin{align*}
& \Big\|\frac 1N\sum_{i=1}^N
(X\sps i - \hat \mu)(X\sps i - \hat \mu)^\top - \eta^2 I\Big\|_\op\\
= {} & \Big\|\frac 1N\sum_{i=1}^N
(X\sps i - \mu)(X\sps i - \mu)^\top - 
(\mu - \hat \mu)(\mu - \hat\mu)^\top - \eta^2 I\Big\|_\op
\tag{$\hat \mu$ is the empirical average of $X\sps i$'s}\\
\leq {} & \|\hat \mu - \mu\|_2^2 + 
\Big\|\frac 1N\sum_{i=1}^N
(X\sps i - \mu)(X\sps i - \mu)^\top - \eta^2 I\Big\|_\op\\
\leq {} & O(N\varepsilon^2) + 
\frac {|G|}{N}\Big\|\frac 1{|G|}\sum_{i\in G}
(X\sps i - \mu)(X\sps i - \mu)^\top - \eta^2 I\Big\|_\op\\
& +
\frac{N - |G|}{N}\Big\|\frac 1{N - |G|}\sum_{i\notin G}
(X\sps i - \mu)(X\sps i - \mu)^\top - \eta^2 I\Big\|_\op\\
\leq {} & O(N\varepsilon^2) + 
\frac{|G|}{N}\cdot \beta^2\\
& +
\frac{N - |G|}{N}\big(\max_{1\leq i\leq N}\|X\sps i - \mu\|_2^2 + \eta^2\big)\\
\leq {} & O(N\varepsilon^2) + \beta^2 + \varepsilon \cdot (O(N) + \eta^2)\\
\leq {} & O(N\varepsilon)\tag{because $\beta,\eta \leq O(1)$ and $\varepsilon \geq 1/N$}.
\end{align*}
Given the inequalities \eqref{eq:QUE-assumption1} and \eqref{eq:QUE-assumption2}, 
Theorem 3.5 and Theorem 4.7 of \citet{dong2019quantum-arxiv} 
complete the proof for $\eta = 0$ and $\eta = 1$, respectively,
with parameters $\gamma_1 = O(\beta\sqrt\varepsilon),
\gamma_2 = O(\beta^2),
\beta_1 = O(\beta/\sqrt\varepsilon),
\beta_2 = O(\beta^2/\varepsilon),
\xi = O(\beta^2),
\kappa = \kappa_1 = O\big(\sqrt N\big),
\kappa_2 = O(N)$.
Note that $\beta_1$ and $\beta_2$ in \citet{dong2019quantum-arxiv} are w.r.t.\ weights
$w$ in $\Delta_{G,2\varepsilon|G|}$, so
we need to translate our goodness conditions on
$\Delta_{G,(1 - 3\varepsilon) N}\supseteq \Delta_{G,(1 - 2\varepsilon)|G|}$
to $\Delta_{G,2\varepsilon|G|}$ using \Cref{claim:pair} below and the triangle inequality.
We leave the details to the readers.
\end{proof}
\begin{claim}
\label{claim:pair}
For all $G\subseteq \{1,\ldots,N\}, 0 < p < |G|$, and $w\in\Delta_{G,p}$ there exists $\tw\in\Delta_{G,|G| - p}$ such that
\begin{equation}
\label{eq:pair}
\frac{p}{|G|}\cdot w + \frac{|G| - p}{|G|}\cdot \tw = w_G,
\end{equation}
where $w_G\sps i = \frac 1{|G|}$ if $i\in G$, and $w_G\sps i = 0$ if $i\notin G$. 
Moreover, the mapping from $w$ to $\tw$ is a bijection from $\Delta_{G,p}$ to $\Delta_{G,|G| - p}$.
\end{claim}
\begin{proof}
We simply solve $\tw$ from \eqref{eq:pair}:
\begin{equation}
\label{eq:wtotw}
\tw = \frac{|G|}{|G| - p}\cdot w_G - \frac{p}{|G| - p}\cdot w.
\end{equation}
It is clear that $\sum_{i=1}^N\tw\sps i = 1$ and $\forall i\notin G, \tw\sps i = 0$. To see that $\tw\in\Delta_{G,|G|-p}$, note that $w\sps i\in [0,\frac 1p]$, so $\tw\sps i = \frac{|G|}{|G| - p}\cdot\frac 1{|G|} - \frac p{|G| - p}\cdot w\sps i\in[0,\frac 1{|G| - p}]$ for all $i\in G$. 
\eqref{eq:wtotw} is clearly injective, and it is also surjective because  
for every $\tw\in\Delta_{G,|G| - p}$, 
we can show that
the inverse mapping below maps it to some $w\in \Delta_{G,p}$, using the same argument above with $p$ replaced by $|G| - p$:
\begin{equation*}
w = \frac{|G|}{p}\cdot w_G - \frac{|G| - p}{p}\cdot \tw.
\end{equation*}
\end{proof}
\section{Proof of Lemma \ref{lm:goodness}}
\label{sec:proof-goodness}
\begin{proof}%
To prove that $X\sps{1,\nu},\ldots,X\sps {N,\nu}$ is 
$(\varepsilon,\eta,\beta')$-good with respect to $\mu'$, 
it suffices to prove the following for all $w\in\Delta_{G, (1 - 3\varepsilon)N}$:
\begin{align}
\Big\|\sum_{i\in G}w\sps i(X\sps{i,\nu} - \mu')\Big\|_2 & \leq \beta'\sqrt\varepsilon;\label{eq:goodness1}\\
\Big\|\big(\sum_{i\in G}w\sps i(X\sps{i,\nu} - \mu')(X\sps{i,\nu} - \mu')^\top\big)  - \eta^2 I\Big\|_\op & \leq (\beta')^2.\label{eq:goodness2}
\end{align}

We first prove \eqref{eq:goodness1}. Since $X\sps{1,0},\ldots,X\sps{N,0}$ is
$(\varepsilon,g_*,\eta,\beta)$-good
with respect to $\mu$
and $G$ satisfies the conditions in \Cref{def:missing}, we have
\begin{equation}
\label{eq:good-mu}
\Big\|\sum_{i\in G}w\sps i(X\sps{i,\mu} - \mu)\Big\|_2 \leq \beta\sqrt\varepsilon.
\end{equation}
We thus bound the $L_2$ norm of the difference
\begin{align*}
\delta := {} & \sum_{i\in G}w\sps i(X\sps{i,\nu} - \mu') - \sum_{i\in G}w\sps i(X\sps{i,\mu} - \mu)\\
= {} & \sum_{i\in G}w\sps i(X\sps{i,\nu} - X\sps{i,\mu}) - (\mu' - \mu).
\end{align*}
By the definition of $X\sps{i,\nu}$ and $X\sps{i,\mu}$ (see \eqref{eq:adjust}), we have 
\begin{equation}
\label{eq:goodness3}
X\sps{i,\nu}_j - X\sps {i,\mu}_j = \left\{
\begin{array}{ll}
0, & \textup{if} ~ i\in G_j,\\
\nu_j - \mu_j, & \textup{otherwise.}
\end{array}
\right.
\end{equation}
Since $\mu'_j := g_j \mu_j + (1 - g_j) \nu_j$, we have 
\begin{equation}
\label{eq:goodness4}
\mu'_j - \mu_j = (1 - g_j)(\nu_j - \mu_j).
\end{equation}
Thus,
the $j$-th coordinate of $\delta$ is 
\begin{equation}
\label{eq:deltaj}
\delta_j = \Big(\sum_{i\in G\backslash G_j}w\sps i - (1 - g_j)\Big)(\nu_j - \mu_j).
\end{equation} 
Since $w\in \Delta_{G, (1 - 3\varepsilon)N}$, we have 
\begin{equation*}
\sum_{i\in G\backslash G_j}w\sps i\leq \frac{1}{(1 - 3\varepsilon)N}\cdot (|G| - |G_j|) \leq \frac 1{(1 - 3\varepsilon)}\cdot (1 - \frac{|G_j|}{|G|})= \frac{1 - g_j}{1 - 3\varepsilon}\leq (1 - g_j) + 3\varepsilon,
\end{equation*}
where we used $g_j\geq g_*\geq 3\varepsilon$ for the last inequality. We also have
\begin{equation*}
\sum_{i\in G\backslash G_j}w\sps i = 1 - \sum_{i\in G_j}w\sps i \geq 1 - \frac{1}{(1 - 3\varepsilon)N}\cdot |G_j| \geq 1 - \frac{1}{1 - 3\varepsilon}\cdot \frac{|G_j|}{|G|} =\frac{1 - g_j - 3\varepsilon}{1 - 3\varepsilon},
\end{equation*}
which implies that $\sum_{i\in G\backslash G_j}w\sps i \geq (1 - g_j) - 3\varepsilon$ because the left-hand-side is always non-negative.
Therefore, \eqref{eq:deltaj} implies that
\begin{equation*}
\|\delta\|_2 \leq 3\varepsilon \|\nu - \mu\|_2 \leq 3\varepsilon\|\nu - \mu\|_g/\sqrt{g_*}\leq O(\rho\sqrt\varepsilon).\tag{$g_*\geq 5\varepsilon$}
\end{equation*}
By \eqref{eq:good-mu} and the triangle inequality, we have
\begin{align*}
\Big\|\sum_{i\in G}w\sps i(X\sps{i,\nu} - \mu')\Big\|_2^2 & \leq (\beta\sqrt\varepsilon + \|\delta\|_2)^2 \\
& \leq (\beta + O(\rho))^2\varepsilon\\
& \leq C_3\, \beta^2\varepsilon + C_4\rho^2\varepsilon. \tag{by AM-GM}
\end{align*}
This proves \eqref{eq:goodness1}. We now turn to proving \eqref{eq:goodness2}.
Define
\begin{equation*}
Z\sps i := (X\sps{i,\nu} - \mu') - (X\sps{i,\mu} - \mu).
\end{equation*}
We can decompose the left-hand-side of \eqref{eq:goodness2} as follows:
\begin{align}
& \sum_{i\in G}w\sps i(X\sps{i,\nu} - \mu')(X\sps{i,\nu} - \mu')^\top - \eta^2 I \notag \\
= {} & \sum_{i\in G}w\sps i(X\sps{i,\mu} - \mu)(X\sps{i,\mu} - \mu)^\top - \eta^2 I \notag \\
& + \sum_{i\in G}w\sps i(X\sps{i,\mu} - \mu)(Z\sps i)^\top 
+ \Big(\sum_{i\in G}w\sps i(X\sps{i,\mu} - \mu)(Z\sps i)^\top \Big)^\top \notag \\
& + \sum_{i\in G}w\sps iZ\sps i(Z\sps i)^\top. \label{eq:goodness5}
\end{align}

By \eqref{eq:goodness3} and \eqref{eq:goodness4}, we have 
\begin{equation}
\label{eq:Z}
Z\sps i_j = \left\{
\begin{array}{ll}
-(1 - g_j)(\nu_j - \mu_j), & \textup{if} ~ i\in G_j,\\
g_j(\nu_j - \mu_j), & \textup{otherwise}.
\end{array}
\right.
\end{equation}
Therefore,
\begin{align*}
& \Big\|\sum_{i\in G} w\sps iZ\sps i(Z\sps i)^\top\Big\|_\op\\
\leq {} & \sum_{i\in G} w\sps i\|Z\sps i\|_2^2\\
= {} & \sum_j\Big(
\sum_{i\in G_j}w\sps i(1 - g_j)^2(\nu_j - \mu_j)^2
+ \sum_{i\in G\backslash G_j}w\sps ig_j^2(\nu_j - \mu_j)^2\Big)\\
\leq {} & \sum_j\Big(
\frac{1}{(1 - 3\varepsilon)N}\sum_{i\in G_j}(1 - g_j)^2(\nu_j - \mu_j)^2
+ \frac{1}{(1 - 3\varepsilon)N}\sum_{i\in G\backslash G_j}g_j^2(\nu_j - \mu_j)^2
\Big)\\
\leq {} & \sum_j\Big(
\frac{1}{1 - 3\varepsilon}(g_j(1 - g_j)^2(\nu_j - \mu_j)^2 + (1 - g_j)g_j^2(\nu_j - \mu_j)^2)
\Big)\\
= {} & \sum_j
\frac {1 - g_j}{1 - 3\varepsilon}
\cdot g_j(\nu_j - \mu_j)^2\\
\leq {} & \|\nu - \mu\|_g^2\\
\leq {} & \rho^2.
\end{align*}
Let $Y_j$ denote the $j$-th column of $\sum_{i\in G}w\sps i(X\sps{i,\mu} - \mu)(Z\sps i)^\top$.
By \eqref{eq:Z}, we have 
\begin{align}
Y_j & = g_j(\nu_j - \mu_j)\sum_{i\in G\backslash G_j}w\sps i(X\sps{i,\mu} - \mu) - (1 - g_j)(\nu_j - \mu_j)\sum_{i\in G_j}w\sps i(X\sps{i,\mu} - \mu).\notag\\
& = g_j(\nu_j - \mu_j)\sum_{i\in G}w\sps i(X\sps{i,\mu} - \mu) - (\nu_j - \mu_j)\sum_{i\in G_j}w\sps i(X\sps{i,\mu} - \mu). \label{eq:Y}
\end{align}
By the $(\varepsilon,g_*,\eta,\beta)$-goodness, we have
\begin{align}
\label{eq:Y2}
& \Big\|\sum_{i\in G}w\sps i(X\sps{i,\mu} - \mu)\Big\|_2\leq \beta\sqrt\varepsilon,\\
\label{eq:Y1}
& \Big\|\sum_{i\in G_j}w\sps i(X\sps{i,\mu} - \mu)\Big\|_2 \le \beta\sqrt\varepsilon, \quad \text{for all }j = 1,\ldots,d.
\end{align}
Plugging \eqref{eq:Y2} and \eqref{eq:Y1} into \eqref{eq:Y}, we have
\begin{align*}
\|Y_j\|_2 & \leq g_j\cdot |\nu_j - \mu_j|\cdot \beta\sqrt\varepsilon + |\nu_j - \mu_j|\cdot \beta\sqrt\varepsilon \\
& \leq O(\beta)\sqrt{\varepsilon/g_*}\cdot \sqrt{g_j}|\nu_j - \mu_j|\tag{$g_j\geq g_*$}\\
& \leq O(\beta)\sqrt{g_j}|\nu_j - \mu_j|.\tag{$g_*\geq 5\varepsilon$}
\end{align*}
Finally, by \eqref{eq:goodness5} and the triangle inequality of the operator norm, we have
\begin{align*}
& \Big\|\sum_{i\in G}w\sps i(X\sps{i,\nu} - \mu')(X\sps{i,\nu} - \mu')^\top - \eta^2 I\Big\|_\op \notag \\
\leq  {} & \Big\|\sum_{i\in G}w\sps i(X\sps{i,\mu} - \mu)(X\sps{i,\mu} - \mu)^\top - \eta^2 I\Big\|_\op \notag \\
& + 2\Big\|\sum_{i\in G}w\sps i(X\sps{i,\mu} - \mu)(Z\sps i)^\top\Big\|_\op \notag \\
& + \Big\|\sum_{i\in G}w\sps iZ\sps i(Z\sps i)^\top\Big\|_\op\\
\leq {} & \beta^2 + 2\Big\|\sum_{i\in G}w\sps i(X\sps{i,\mu} - \mu)(Z\sps i)^\top\Big\|_F + \rho^2\tag{$\|\cdot\|_F$ denotes the Frobenius norm}\\
= {} & \beta^2 + 2\big(\sum_{j=1}^d\|Y_j\|_2^2\big)^{1/2} + \rho^2\\
\leq {} & \beta^2 + 2\cdot O(\beta \rho) + \rho^2\\
\leq {} & C_3\beta^2 + C_4\rho^2.\tag{by AM-GM}
\end{align*}
This proves \eqref{eq:goodness2} and hence the lemma.
\end{proof}
\section{Proof of Lemma \ref{lm:subgaussian-goodness}}
\label{sec:subgaussian-goodness}
\begin{definition}
\label{def:uncorrupted}
We use the term \emph{uncorrupted dataset} to denote the hypothetical dataset
if the adversary in \Cref{sec:model} skipped step 3. 
\end{definition}
\Cref{lm:subgaussian-goodness} is a direct corollary of the following lemma
by excluding the $\leq \varepsilon N$ corrupted examples from $G$. 
\begin{lemma}
\label{thm:Gaussian-concentration}
Let $X\sps 1,\ldots,X\sps N\in\bR^d$ be an \emph{uncorrupted} dataset generated with $\cP = \cP_1(\eta)$ and mean $\mu$.
Fill every missing coordinate by independent $\mathcal N(0,1)$ variables 
to form the completed dataset $X\sps{1,0},\ldots,X\sps{N,0}$.
Define $G = \{1,\ldots,N\}$
and let
$\delta\in(0,1/2), \varepsilon \in (0, 1/10]$.
Then, for some $\beta = O(\sqrt{(d + \log (1/\delta))/ N\varepsilon + \varepsilon\log(1/\varepsilon)})$,
with probability at least $1- \delta$,
the following holds
for all $w\in \Delta_{G,(1 - 3\varepsilon)N}$:
\begin{align}
\label{eq:gaussian-concentration-goal-1}
& \Big\|\sum_{i = 1}^Nw\sps i(X\sps{i,\mu} - \mu)\Big\|_2 \leq \beta\sqrt\varepsilon,\\
\label{eq:gaussian-concentration-goal-2}
& \Big\|\Big(\sum_{i = 1}^Nw\sps i(X\sps{i,\mu} - \mu)(X\sps{i,\mu} - \mu)^\top\Big) - \eta^2 I\Big\|_\op\leq \beta^2,\\
\label{eq:gaussian-concentration-goal-3}
& \Big\|\sum_{i\in \Gamma_j}w\sps i(X\sps{i,\mu} - \mu)\Big\|_2 \leq \beta\sqrt\varepsilon, \quad \text{for all } j = 1,\ldots, d.
\end{align}
\end{lemma}

\begin{proof}%
By \Cref{claim:composition}, every $X\sps{i,\mu}$ can be treated as drawn independently from some distribution in $\cP_1(\eta)$ with mean $\mu$.

Inequalities \eqref{eq:gaussian-concentration-goal-1} and \eqref{eq:gaussian-concentration-goal-2} follow from similar arguments to \citet[Lemma 2.1.4 - 2.1.9]{li2018principled}. To establish \eqref{eq:gaussian-concentration-goal-3}, we consider a new dataset $Y\sps{1},\ldots,Y\sps{N}$, where
\begin{equation*}
Y\sps{i} = \left\{\begin{array}{ll} X\sps{i,\mu}, & \text{if }i\in \Gamma_j;\\ \mu, & \text{if } i\notin \Gamma_j.\end{array}\right.
\end{equation*}
Every $Y\sps i$ is drawn independently from a sub-Gaussian distribution with mean 
$\mu$ and covariance PSD-dominated by $I$. By similar arguments to  \citet[Lemma 2.1.9]{li2018principled}, we have
\[
\forall w\in \Delta_{G,(1 - 3\varepsilon N)}, \Big\|\sum_{i\in \Gamma_j}w\sps i(X\sps{i,\mu} - \mu)\Big\|_2 = \Big\|\sum_{i=1}^N w\sps i (Y\sps i - \mu)\Big\|_2 \leq \beta\sqrt\varepsilon.
\] 
\eqref{eq:gaussian-concentration-goal-3} follows from a union bound over $j = 1,\ldots,d$. Note that $d + \log(d/\delta) = O(d + \log(1/\delta))$, so we can afford such a union bound with $\beta$ increasing by no more than a constant factor.
\end{proof}
\section{Proof of Lemma \ref{lm:bounded-variance-goodness}}
\label{sec:bounded-variance-goodness}
Recall the definition of the uncorrupted dataset in \Cref{def:uncorrupted}.
\Cref{lm:bounded-variance-goodness} is a direct corollary of the following lemma
by excluding the $\leq \varepsilon N$ corrupted examples from $G$.
\begin{lemma}
\label{thm:variance-concentration}
Let $X\sps 1,\ldots,X\sps N\in\bR^d$ be an \emph{uncorrupted} dataset generated with $\cP = \cP_2$ and mean $\mu$.
Fill every missing coordinate by zeros
to form the completed dataset $X\sps{1,0},\ldots,X\sps{N,0}$.
Let $\delta\in (0, 1/2), \varepsilon\in (0, 1/11]$ and assume $\varepsilon N \geq \Omega(\log(d/\delta))$.
Then, for some $C = O(1)$ and $\beta = O(\sqrt{d\log (d/\delta)/N\varepsilon + 1})$,
with probability at least $1-\delta$,
the following holds for 
$G:=\{i\in\{1,\ldots,N\}:\|X\sps{i,\mu} - \mu\|_2 \leq C\sqrt{d/\varepsilon}\}$:
\begin{OneLiners}
\item $|G|\geq (1-0.1\varepsilon)N$;
\item For all $w\in \Delta_{G,(1 - 3\varepsilon')N}$ where $\varepsilon' = 1.1\varepsilon$, 
\begin{align}
\label{eq:variance-concentration-goal-1}
& \Big\|\sum_{i = 1}^Nw\sps i(X\sps {i,\mu} - \mu)\Big\|_2 \leq \beta\sqrt\varepsilon,\\
\label{eq:variance-concentration-goal-2}
& \Big\|\sum_{i = 1}^Nw\sps i(X\sps {i,\mu} - \mu)(X\sps {i,\mu} - \mu)^\top \Big\|_\op\leq \beta^2,\\
\label{eq:variance-concentration-goal-3}
& \Big\|\sum_{i\in \Gamma_j}w\sps i(X\sps{i,\mu} - \mu)\Big\|_2 \leq \beta\sqrt\varepsilon,\quad\text{for all }j=1,\ldots,d
\end{align}
\end{OneLiners}
\end{lemma}
We first prove two helper claims.
\begin{claim}
\label{claim:variance-helper}
Assume $X\in\bR^d$ is drawn from distribution $\cD\in\cP_2$ with mean $\mu$. Let $\cE$ denote an event
which contains the event that 
$\|X - \mu\|_2\leq C\sqrt{d/\varepsilon}$ for some $\varepsilon\in(0,1/2)$ and $C > 1$. Let $\ \tD$ denote the distribution of $X$ conditioned on $\cE$. We have
\begin{OneLiners}
\item $\Pr_\cD[\cE] \geq 1 - C^{-2}\varepsilon$;
\item $\|\bE_{\tX\sim\tD}[\tX] - \mu\|_2 \leq \sqrt\varepsilon$;
\item $\forall v, \bE_{\tX\sim\tD}[((\tX - \mu)^\top v)^2]\leq \frac {1}{1 - \varepsilon}\|v\|_2^2$.
\end{OneLiners}
\end{claim}
\begin{proof}

Since $\cD$ has covariance $\Sigma \preceq I$, it has variance at most one in every direction. In particular, $\bE_{X\sim \cD}[(X_j - \mu_j)^2] \leq 1$ for all coordinate $j$ and thus $\bE_{X\sim\cD}[\|X - \mu\|_2^2]\leq d$. By Markov's inequality, we have $\Pr_\cD[\cE]\geq 1 - C^{-2}\varepsilon$. If $\Pr_\cD[\cE] = 1$, then $\tD = \cD$ and the claim holds trivially. We assume henceforth that $\Pr_{\cD}[\cE] < 1$.

Let $\neg\cE$ denote the complement of $\cE$, so ${\Pr}_\cD[\neg \cE] = 1 - {\Pr}_\cD[\cE] \leq C^{-2}\varepsilon$. 
Let $v$ denote $\bE_{X\sim\cD}[X|\neg\cE]\allowbreak - \mu$. We have
\begin{align*}
 \|v\|_2^2 & \geq \bE_{X\sim\cD}[((X - \mu)^\top v)^2]\\
& \geq {\Pr}_{\cD}[\neg \cE]\cdot \bE_{X\sim\cD}[((X - \mu)^\top v)^2|\neg \cE]\\
& \geq {\Pr}_{\cD}[\neg \cE]\cdot ((\bE_{X\sim\cD}[X|\neg \cE] - \mu)^\top v)^2\tag{Jensen's inequality}\\
& = {\Pr}_{\cD}[\neg \cE]\cdot \|v\|_2^4.
\end{align*}
Therefore, $\|v\|_2\leq 1 / \sqrt{{\Pr}_{\cD}[\neg \cE]}$. Note that
\begin{equation*}
{\Pr}_\cD[\cE]\cdot (\bE_{\tX\sim\tD}[\tX] - \mu) + {\Pr}_{\cD}[\neg \cE]\cdot v = \bE_{X\sim\cD}[X] - \mu = 0.
\end{equation*}
Therefore, 
\begin{equation*}
\|\bE_{\tX\sim\tD}[\tX] - \mu\|_2 = \frac{{\Pr}_{\cD}[\neg \cE]}{{\Pr}_{\cD}[\cE]}\cdot \|v\|_2 \leq \frac{{\Pr}_{\cD}[\neg \cE]}{{\Pr}_{\cD}[\cE]}\cdot \frac{1}{\sqrt{{\Pr}_\cD[\neg\cE]}}\leq \sqrt\varepsilon.
\end{equation*}
Finally, for a generic vector $v$, we have
\begin{align*}
\|v\|_2^2 & \geq \bE_{X\sim \cD}[((X - \mu)^\top v)^2]\\
& \geq {\Pr}_{\cD}[\cE]\cdot \bE_{\tX\sim\tD}[((\tX - \mu)^\top v)^2],
\end{align*}
which implies $\bE_{\tX\sim\tD}[((\tX - \mu)^\top v)^2]\leq \frac 1{{\Pr}_\cD[\cE]}\cdot \|v\|_2^2 \leq \frac{1}{1 - \varepsilon} \|v\|_2^2$.
\end{proof}
\begin{claim}
\label{claim:variance-helper2}
Suppose we have examples $X\sps 1,\ldots,X\sps m\in\bR^d$ where every $X\sps i$ is drawn independently from $\cD\sps i$. Assume there is a vector $\mu$ such that for all $i$, 
$\|X\sps i - \mu\|_2 \leq O(\sqrt{d/\varepsilon})$ almost surely and $\bE[\|X\sps i - \mu\|_2^2] \leq O(d)$. 
Assume $m\varepsilon \geq \Omega(\log(1/\delta))$ for $\delta\in(0,1/2)$.
Then with probability at least $1 - \delta$,
\begin{equation*}
\Big\|\frac 1{m}\sum_{i=1}^m(X\sps i - \bE[X\sps i])\Big\|_2 \leq 
O(\sqrt{d\log (1/\delta)/m}).
\end{equation*} 
\end{claim}
\begin{proof}
Define $\mu\sps i := \bE[X\sps i]$.
Consider the function 
\begin{equation*}
f(X\sps{1},\ldots,X\sps{m}) := \Big\|\frac 1{m}\sum_{i = 1}^m(X\sps{i} - \mu\sps i)\Big\|_2.
\end{equation*}
Choose an arbitrary $i\in\{1,\ldots,m\}$ and fix all the $X\sps{i'}$ for $i'\neq i$.
Treating $f$ as a function of $X\sps{i}$, we have the following by the triangle inequality:
\begin{equation*}
|f(X\sps{i}) - f(\mu)| \leq \frac 1{m}\|X\sps {i} - \mu\|_2.
\end{equation*}
Thus, $|f(X\sps{i}) - f(\mu)|\leq O((\sqrt{d/\varepsilon})/m)$
almost surely, and $\bE[(f(X\sps i) - f(\mu))^2]\leq O(d/m^2)$.
Applying McDiarmid’s inequality of Bernstein type to $f$ (see \citet{maurer2019bernstein} and references therein), the following holds with probability at least $1-\delta$:
\begin{equation}
\label{eq:mcdiarmid}
f(X\sps 1,\ldots,X\sps m) \le \bE[f(X\sps 1,\ldots,X\sps m)] + O(\sqrt{d\log (1/\delta)/m}).
\end{equation}
Now we bound $\bE[f(X\sps 1,\ldots,X\sps m)]$. By Jensen's inequality,
\begin{align}
\bE[f(X\sps 1,\ldots,X\sps m)]^2 & \le \bE[f(X\sps 1,\ldots,X\sps m)^2]\notag\\
& = \frac{1}{m^2}\bE\left[\Big\|\sum_{i = 1}^m(X\sps{i} - \mu\sps i)\Big\|_2^2\right]\notag\\
& = \frac{1}{m^2}\sum_{1\le i,j\le m}\bE[(X\sps i - \mu\sps i)^\top (X\sps j - \mu\sps j)]\notag\\
& = \frac{1}{m^2}\sum_{i = 1}^m\bE[\|X\sps i - \mu\sps i\|_2^2]\notag\\
& \le \frac{1}{m^2}\sum_{i = 1}^m\bE[\|X\sps i - \mu\|_2^2]\notag\\
& \le \frac{d}{m}.\label{eq:mcdiarmid-jensen}
\end{align}
Plugging \eqref{eq:mcdiarmid-jensen} into \eqref{eq:mcdiarmid},
with probability at least $1-\delta$,
\[
f(X\sps 1,\ldots,X\sps m) \le O(\sqrt{d/m}) + O(\sqrt{d\log (1/\delta)/m}) = O(\sqrt{d\log (1/\delta)/m}).\qedhere
\]
\end{proof}
\begin{proof}[Proof of \Cref{thm:variance-concentration}]
By \Cref{claim:composition}, every $X\sps{i,\mu}$ can be treated as drawn independently from some $\cD\sps i\in\cP_2$ with mean $\mu$. By the first guarantee of \Cref{claim:variance-helper}, every $i=1,\ldots,N$ belongs to $G$ independently with probability at least $1 - C^{-2}\varepsilon$. 
Given $\varepsilon N \geq \Omega(\log(1/\delta))$, 
we choose $C$ to be a sufficiently large constant so that
by the Chernoff bound, 
with probability $1 - \delta/4$ we have 
\begin{equation}
\label{eq:large-good-set}
|G|  \geq (1 - 0.1\varepsilon) N. 
\end{equation}
We now conditioned on $G$ satisfying the inequality above.
For every $i\in G$, $X\sps{i,\mu}$ can be treated as drawn independently from $\widetilde\cD\sps i$, defined as the conditional distribution of $\cD\sps i$ on the event $\|X\sps {i,\mu} - \mu\|_2 \leq C\sqrt{d/\varepsilon}$. 
The third guarantee of \Cref{claim:variance-helper} implies $\bE[\|X\sps {i,\mu} - \mu\|_2^2] = \sum_{j=1}^d\bE[(X\sps{i,\mu}_j - \mu_j)^2] \leq \frac{d}{1 - \varepsilon}$. Then by \Cref{claim:variance-helper2}, 
with probability $1 - \delta/4$ we have
\begin{align}
\Big\|\frac 1{|G|}\sum_{i\in G}(X\sps{i,\mu} - \mu)\Big\|_2  \leq O(\beta\sqrt\varepsilon),\label{eq:variance-concentration5}
\end{align}
where we used $\|\bE[X\sps{i,\mu}] - \mu\|_2\leq O(\beta\sqrt\varepsilon)$ due to the second guarantee of \Cref{claim:variance-helper}.

Applying the same argument to $Y\sps{1},\ldots,Y\sps{N}$ where
\begin{equation*}
Y\sps{i} = \left\{\begin{array}{ll} X\sps{i,\mu}, & \text{if }i\in \Gamma_j;\\ \mu, & \text{if } i\notin \Gamma_j.\end{array}\right.,
\end{equation*}
with probability at least $1 - \delta/4$,
\begin{equation}
\Big\|\frac 1{|G|}\sum_{i\in G\cap\Gamma_j}(X\sps{i,\mu} - \mu)\Big\|_2
= \Big\|\frac 1{|G|}\sum_{i\in G}(Y\sps{i} - \mu)\Big\|_2
\leq O(\beta\sqrt\varepsilon),\quad\text{for all }j = 1,\ldots,d.\label{eq:variance-concentration6}
\end{equation}

{\par \sloppy \hbadness 10000
The third guarantee of \Cref{claim:variance-helper} implies
\begin{equation*}
\big\|\bE[(X\sps{i,\mu} - \mu)(X\sps{i,\mu} - \mu)^\top]\big\|_\op \leq \frac{1}{1-\varepsilon}\leq O(1).
\end{equation*}
We also have $\big\|(X\sps{i,\mu} - \mu)(X\sps{i,\mu} - \mu)^\top\big\|_\op = \|X\sps{i,\mu} - \mu\|_2^2 \leq O(d/\varepsilon)$.
By the matrix Chernoff bound \citep[Theorem 5.1.1]{tropp2015introduction}, with probability $1 - \delta/4$ we have
\begin{equation}
\label{eq:variance-concentration3}
\Big\|\frac 1{|G|}\sum_{i\in G}(X\sps{i,\mu} - \mu)(X\sps {i,\mu} - \mu)^\top\Big\|_\op
\leq O(\beta^2).
\end{equation}
Clearly, inequality \eqref{eq:variance-concentration3} implies the following:
\begin{equation}
\label{eq:variance-concentration3-1}
\Big\|\frac 1{|G|}\sum_{i\in G\cap \Gamma_j}(X\sps{i,\mu} - \mu)(X\sps {i,\mu} - \mu)^\top\Big\|_\op
\leq O(\beta^2),\quad\text{for all }j = 1,\ldots,d.
\end{equation}
Below, we will show that \eqref{eq:variance-concentration5} and \eqref{eq:variance-concentration3} imply \eqref{eq:variance-concentration-goal-1} and \eqref{eq:variance-concentration-goal-2}. A similar argument which we omit here shows that \eqref{eq:variance-concentration6} and \eqref{eq:variance-concentration3-1} imply \eqref{eq:variance-concentration-goal-3}, completing the proof.

By \Cref{claim:pair}, for the bijection between $w\in \Delta_{G,|G| - 3\varepsilon' N}$ and $\tw\in\Delta_{G,3\varepsilon' N}$, we have
\begin{equation}
\label{eq:variance-concentration2}
\frac{|G| - 3\varepsilon' N}{|G|}\cdot w + \frac{3\varepsilon' N}{|G|}\cdot \widetilde w = w_G.
\end{equation}
Therefore,
for any vector $v$,
\begin{align}
O(\beta^2\|v\|_2^2) & \geq \frac 1{|G|}\sum_{i\in G}((X\sps{i,\mu} - \mu)^\top v)^2 
\tag{by \eqref{eq:variance-concentration3}}\\
& = \frac{|G| - 3\varepsilon' N}{|G|}\sum_{i\in G}w\sps i((X\sps{i,\mu} - \mu)^\top v)^2 \notag\\
& ~~~ + \frac{3\varepsilon' N}{|G|}\sum_{i\in G}\tw \sps i((X\sps{i,\mu} - \mu)^\top v)^2.\label{eq:variance-concentration31}
\end{align}
This implies that for all $w\in\Delta_{G,(1 - 3\varepsilon')N}\subseteq \Delta_{G,|G| - 3\varepsilon'N}$, we have
\begin{equation}
\label{eq:variance-concentration1}
\Big\|\sum_{i\in G}w\sps i(X\sps{i,\mu} - \mu)(X\sps{i,\mu} - \mu)^\top\Big\|_\op \leq \frac{|G|}{|G| - 3\varepsilon' N}\cdot O(\beta^2) = O(\beta^2).
\end{equation}
This proves \eqref{eq:variance-concentration-goal-2}.
Choosing $v= \sum_{i\in G}\tw\sps i(X\sps{i,\mu} - \mu)$ in \eqref{eq:variance-concentration31}, we have
\begin{align*}
O(\beta^2\|v\|_2^2) & \geq \frac{3\varepsilon' N}{|G|}\sum_{i\in G}\tw\sps i((X\sps{i,\mu} - \mu)^\top v)^2\\
& \geq \frac{3\varepsilon' N}{|G|}(\sum_{i\in G}\tw\sps i (X\sps{i,\mu} - \mu)^\top v)^2\tag{Jensen's inequality}\\
& = \frac{3\varepsilon' N}{|G|}\|v\|_2^4.
\end{align*}
This implies
\begin{equation}
\label{eq:variance-concentration4}
\forall \tw\in\Delta_{G,3\varepsilon'N},\Big\|\sum_{i\in G}\tw\sps i ( X\sps{i,\mu} - \mu)\Big\|_2 = O(\beta / \sqrt\varepsilon).
\end{equation}
From \eqref{eq:variance-concentration2}, we also have
\begin{equation*}
\frac{|G| - 3\varepsilon'N}{|G|}\sum_{i\in G}w\sps i(X\sps{i,\mu} - \mu)
+
\frac{3\varepsilon'N}{|G|}\sum_{i\in G}\tw \sps i(X\sps{i,\mu} - \mu)
=
\frac 1{|G|}\sum_{i\in G}(X\sps {i,\mu} - \mu).
\end{equation*}
Combining this with \eqref{eq:variance-concentration5} and \eqref{eq:variance-concentration4} using the triangle inequality, for all $w\in\Delta_{G,(1 - 3\varepsilon') N}\subseteq \allowbreak \Delta_{G, |G| - 3\varepsilon' N}$, we have
\begin{equation}
\label{eq:variance-concentration7}
\Big\|\sum_{i\in G} w\sps i (X\sps{i,\mu} - \mu) \Big\|_2 \leq O(\beta\sqrt\varepsilon).
\end{equation}
This proves \eqref{eq:variance-concentration-goal-1}.
}
\end{proof}
\section{Proof of Claim \ref{claim:coordinate-wise-median}}
\label{sec:proof-median}

\begin{proof}
Let $X\sps 1,\ldots, X\sps N$ denote the uncorrupted dataset as in \Cref{def:uncorrupted}.
Define $N_j$ as the number of examples with $X\sps i_j \neq *$. We condition on the event that $N_j\geq (\gamma-\varepsilon)N$ because this is a necessary condition for
the corresponding $\varepsilon$-corrupted dataset
to be $\gamma$-complete.
Let $C > 1$ be a sufficiently large constant.
Define $N_{j,\leq}$ as the number of examples with $X\sps i_j\neq *$ that satisfy $|X\sps i_j - \mu_j|\leq C$, and define $N_{j,>}$ as the number of examples with $X\sps i_j\neq *$ that satisfy $|X\sps i_j - \mu_j| > C$. Clearly, $N_{j,\leq} + N_{j,>} = N_j$. 

As long as $X\sps i_j\neq *$, we have
$\Pr[|X\sps i _j - \mu_j|\geq C]\leq C^{-2}$ by Chebyshev's inequality. 
Since we assumed $\gamma N \geq \Omega(\log(d/\delta))$, by the Chernoff bound, choosing $C$ sufficiently large guarantees that 
for every fixed $j$,
with probability at least $1 - \delta/d$,
it happens that
at least $2/3$ of the examples with $X\sps i_j\neq *$ satisfy $|X\sps i_j - \mu_j|\leq C$.
By a union bound,
this happens
for all $j = 1,\ldots,d$
with probability at least $1 - \delta$,
in which case 
\begin{equation}
\label{eq:median1}
N_{j,\leq} - N_{j, >} \geq (1/3)N_j \geq (1/3)(\gamma - \varepsilon)N > 2\varepsilon N,\quad\textup{for all}~j = 1,\ldots,d.
\end{equation}
Since the adversary can corrupt at most $\varepsilon N$ examples, \eqref{eq:median1} implies that there are more examples satisfying $|X\sps i_j - \mu_j|\leq C$ than ones satisfying $|X\sps i_j - \mu_j|> C$ in the $\varepsilon$-corrupted dataset, and thus the coordinate-wise median $\nu$ satisfies $|\nu_j - \mu_j|\leq C$ for all $j$.
This completes the proof for $\|\nu - \mu\|_2 \leq O\big(\sqrt d\big)$.
\end{proof}
\section{Omitted Analysis for Hashing}
We prove the following claim showing that hashing does not
effectively eliminate the missing entries.
\label{sec:proof-hashing}
\begin{claim}
\label{claim:hashing-lower}
Given a dataset $X\sps 1,\ldots,X\sps N\in (\bR\cup\{*\})^d$ 
where every coordinate $j$
is present in exactly $\gamma N$ examples chosen uniformly at random
and independently for every $j$.
Then, with probability at least $1 - 2\exp(-\gamma N / 8)$,
no hash function $h$ gives a new dataset of size $N' = \gamma N/C$ that contains
less than $1/e^{4C+2}$-fraction of missing entries, assuming 
$\gamma \leq 1/4$, 
$N \geq (32C + 16)/\gamma$, and 
$d \geq e^{4C}C(8(\log N')/\gamma + 1)$.
\end{claim}
\begin{proof}
It is more convenient to choose the present entries in a different way, 
where every entry of the data matrix is present independently with probability $2\gamma$.
Of course, this may not always give a $\gamma$-complete dataset,
and we say a coordinate $j$ is ``bad''
if there are fewer than $\gamma N$ examples with coordinate $j$ present.
By the Chernoff bound and the linearity of expectation,
the expected fraction of 
``bad'' coordinates is at most $\exp(-\gamma N/4)$.
Therefore, by the Markov bound,
with probability at least 
$1 - \exp(4C + 2 -\gamma N/4) \geq 1 - \exp(-\gamma N / 8)$, 
at most $1/e^{4C + 2}$ fraction of the coordinates are ``bad''.
To repair each ``bad'' coordinate $j$, 
we can re-choose the $\gamma N$ examples with coordinate $j$ present
uniformly at random.
Thus, with probability at least $1 - \exp(-\gamma N/8)$,
repairing all the ``bad'' coordinates 
decreases the fraction of missing entries in the new dataset by
at most $1/e^{4C+2}$ regardless of the hash function $h$.

Now we assume that we do not repair the ``bad'' coordinates,
and it suffices to show that in this case
with probability at least $1 - \exp(-\gamma N/8)$, 
no hash function gives a new dataset with less than $1/e^{4C + 1}$-fraction
of missing entries.

Let us fix a hash function $h$, 
for which we lower bound the probability that
the fraction of missing entries in the new dataset is at least $1/e^{4C+1}$. 
After that, we apply the union bound over all hash functions.
For a fixed hash function $h$, 
the expected fraction of missing entries in the new dataset is
\begin{equation*}
\frac{1}{N'}\sum_{i'=1}^{N'}(1 - 2\gamma)^{|h^{-1}(i')|}
\geq (1 - 2\gamma)^{N/N'}\geq e^{-4C},
\end{equation*}
where we used the AM-GM inequality and the fact that $(1 - 2\gamma)^{1/\gamma}\geq e^{-4}$ whenever $0<\gamma \leq 1/4$.
Moreover, each of the $dN'$ entries in the new dataset is missing independently,
so the Chernoff bound implies that
with probability at least $1 - \exp(-dN'e^{-4C}/8)$, 
the fraction of missing entries in the new dataset is at least $1/e^{4C+1}$.
Finally,
by the union bound over all $(N')^N$ hash functions,
with probability at least
\begin{align*}
1 - (N')^N \exp(-dN'e^{-4C}/8) & = 1 - \exp\big(
-(\gamma N/8)(
de^{-4C}/C-8(\log N')/\gamma
)
\big)\\
&\geq 
1 - \exp(-\gamma N / 8),
\end{align*}
no hash function gives a new dataset 
with less than $1/e^{4C+1}$-fraction of missing entries.
\end{proof}
\end{document}